\documentclass[11pt,envcountsame]{article}
\usepackage{amsmath,amssymb,amsthm,thmtools}

\usepackage{graphicx}
\usepackage{color}
\usepackage{latexsym}
\usepackage{csquotes}
\usepackage{verbatim}
\usepackage{tikz}
\usepackage{tabularx}
\usepackage{bm}
\usepackage{framed}
\usepackage{mathtools}
\usepackage{microtype}
\usepackage{authblk}
\usepackage{url}

\usepackage{fullpage}

\newtheorem{definition}{Definition}

\newtheorem{lemma}[definition]{Lemma}
\newtheorem{theorem}[definition]{Theorem}

\newtheorem{example}[definition]{Example}

\usetikzlibrary{arrows,shapes.gates.logic.US,shapes.gates.logic.IEC,calc}

\mathchardef\mhyphen="2D

\newcommand{\eq}{\ensuremath{\mathrm{Eq}}}

\newcommand{\T}{\ensuremath{T}}
\newcommand{\F}{\ensuremath{F}}

\newcommand{\neqq}{\mathrm{Ne}}

\newcommand{\clone}[1]{\ensuremath{[#1]}}

\newcommand{\cclone}[1]{\ensuremath{\langle #1 \rangle}}
\newcommand{\pol}{\ensuremath{\mathrm{Pol}}}
\newcommand{\ppol}{\ensuremath{\mathrm{pPol}}}

\newcommand{\inv}{\ensuremath{\mathrm{Inv}}}
\newcommand{\pcclone}[1]{\ensuremath{\langle #1 \rangle_{\not
      \exists}}}
\newcommand{\pro}{\ensuremath{\mathrm{Proj}}}
\newcommand{\cc}[3]{\ensuremath{\cloneFont{{#1}}^{#2}_{#3}}}

\newcommand{\nandn}[2]{\ensuremath{\mathrm{NA}^{#2}_{\scriptscriptstyle
      #1}}}
\newcommand{\oddn}[2]{\ensuremath{#2\mhyphen\mathrm{ODD}^{\scriptscriptstyle
      #1}}}

\newcommand{\evenn}[2]{\ensuremath{#2\mhyphen\mathrm{EVEN}^{\scriptscriptstyle
      #1}}}

\newcommand{\ra}{\ensuremath{n}}

\newcommand{\overbar}[1]{\mkern
  0.8mu\overline{\mkern-0.8mu#1\mkern-0.8mu}\mkern 0.8mu}

\newcommand{\cloneFont}[1]{\mathsf{#1}}
\newcommand{\ar}{\ensuremath{\mathrm{ar}}}
\newcommand{\domain}{\ensuremath{\mathrm{dom}}}

\newcommand{\zeroes}{\ensuremath{\bm{0}}}
\newcommand{\ones}{\ensuremath{\bm{1}}}

\newcommand{\cols}{\ensuremath{U}}

\newcommand{\functions}[1]{\ensuremath{\cloneFont{OP}_{#1}}}
\newcommand{\pfunctions}[1]{\ensuremath{\cloneFont{pOP}_{#1}}}
\newcommand{\relations}[1]{\ensuremath{\cloneFont{REL}_{#1}}}

\newcommand{\problemFont}[1]{\textsc{#1}}
\newcommand{\CSP}{\protect\ensuremath\problemFont{CSP}}
\newcommand{\SAT}{\protect\ensuremath\problemFont{SAT}}

\newcommand{\USAT}{\protect\ensuremath{\problemFont{U-SAT}}}
\newcommand{\UCSP}{\protect\ensuremath{\problemFont{U-CSP}}}
\newcommand{\UKSAT}[1]{\protect\ensuremath{\problemFont{UNIQUE-$#1$-SAT}}}
\newcommand{\UNSAT}{\protect\ensuremath{\problemFont{UNSAT}}}
\newcommand{\UNCSP}{\protect\ensuremath{\problemFont{UNCSP}}}

\newcommand{\Rdddp}{\ensuremath{R^{\scriptscriptstyle \neq \neq \neq 
    0 1}_{\scriptscriptstyle 1/3}}}

\begin{document}

\title{On the Strength of Uniqueness Quantification in Primitive
  Positive Formulas}

\author[1]{Victor
  Lagerkvist\thanks{victor.lagerkvist@liu.se}}
\author[2]{Gustav Nordh\thanks{gustav.nordh@gmail.com}}

\affil[1]{\small Department of Computer and Information Science, Link\"oping
  University, Link\"oping, Sweden}
\affil[2]{\small Independent researcher}

\date{}

\maketitle


\abstract{
  Uniqueness quantification ($\exists !$) is a quantifier
  in first-order logic where one requires that exactly one element
  exists satisfying a given property. 
  In this paper we investigate the strength of uniqueness
  quantification when it is used in place of existential
  quantification in conjunctive
  formulas over a given set of relations $\Gamma$, so-called {\em primitive positive
    definitions} (pp-definitions). We fully classify the Boolean sets of relations
  where uniqueness quantification has the same strength as existential
  quantification in pp-definitions and give several results valid for arbitrary finite
  domains. 
  We also consider applications of $\exists !$-quantified pp-definitions in
  computer science, which can be used to study the computational
  complexity of problems where the number of solutions is
  important. Using our classification we give a new and
  simplified proof of the trichotomy theorem for the unique
  satisfiability problem, and prove a general result for the unique
  constraint satisfaction problem. 
  Studying these problems in a more rigorous framework also
  turns out to be advantageous in the context of lower bounds, and we
  relate the complexity of these problems to the {\em
    exponential-time hypothesis}.
}
\section{Introduction}
\label{sec:intro}
A {\em primitive positive definition} (pp-definition) over a
relational structure ${\cal A} = (A; R_1, \ldots, R_k)$ is a first-order formula $\exists y_1,
\ldots, y_m \colon \varphi(x_1, \ldots, x_n, y_1, \ldots, y_m)$ with
free variables $x_1, \ldots, x_n$ where
$\varphi(x_1, \ldots, x_n, y_1, \ldots, y_m)$ is a conjunctive
formula. Primitive positive definitions have been
extremely influential in the last decades due to their one-to-one
correspondence with term algebras in universal algebra, making them a
cornerstone in the {\em algebraic approach} for studying
computational complexity~\cite{barto2017,creignou2008b}. 
In short, pp-definitions can be used to obtain
classical ``gadget reductions'' between problems by
replacing constraints by their pp-definitions, which in the process
might introduce fresh variables viewed as being existentially quantified. This approach has
successfully been used to study the complexity of e.g.\ the {\em
constraint satisfaction problem} (CSP) which recently led to a
dichotomy between tractable and NP-complete CSPs~\cite{bulatov2017,zhuk2017}. However, these reductions are
typically not sufficient for optimisation problems and other variants
of satisfiability, where one needs reductions preserving the
number of models, so-called {\em parsimonious} reductions. Despite the
tremendous advances in the algebraic approach there is
currently a lack of methods for studying problems
requiring parsimonious reductions, and in this paper we take the first
step in developing such a framework.
The requirement of parsimonious reductions can be realised by
restricting existential quantification to {\em unique quantification}
($\exists !$), where we explicitly require that the variable in question can be
expressed as a unique combination of other variables. 
That is, ${\cal A} \models \exists ! x_i \colon \varphi(x_1, \ldots, x_i,
\ldots, x_n)$ if and only if there exists a function $f$ such
that $f(a_1, \ldots, a_{i-1}, a_{i+1}, \ldots, a_n) = a_i$, for all $a_1,
\ldots, a_{i-1}, a_i, a_{i+1}, \ldots, a_n \in A$ where ${\cal A}
\models \varphi(a_1, \ldots, a_{i-1}, a_i, a_{i+1}, \ldots,
a_n)$. This notion of unique  quantification is not the
only one possible and we discuss an alternative viewpoint in Section~\ref{sec:conclusions}.
As a first step in understanding the applicability of uniqueness quantification
in complexity classifications we are interested in studying
the expressive power of unique existential quantification when used in
place of existential quantification in pp-definitions, which we call
{\em upp-definitions}. Any variables introduced by the
resulting gadget reductions are then uniquely determined and do
not affect the number of models.

Our main question is then: for
which relational structures $\mathcal{A}$ is it the case that 
for every pp-formula $\varphi(x_1, \ldots, x_n)$ there exists a upp-formula
$\vartheta(x_1, \ldots, x_n)$ such that
$\mathcal{A} \models \varphi(a_1, \ldots, a_n) \Leftrightarrow
\mathcal{A} \models \vartheta(a_1, \ldots, a_n)$ for all $a_1,
\ldots, a_n \in A$? If this holds over ${\mathcal{A}}$ then uniqueness
quantification has the same expressive power as existential
quantification. The practical motivation for studying this is that if
upp-definitions are as powerful as pp-definitions in $\mathcal{A}$,
then, intuitively, any gadget reduction between two problems can be
replaced with a parsimonious reduction.
Given the generality of this question a complete answer for arbitrary
relational structures is well out of reach, and we begin by introducing simplifying concepts.
First, 
pp-definitions can be viewed as a closure operator over
relations, and the resulting closed sets of relations are
known as {\em relational clones}, or {\em co-clones}~\cite{lau2006}. For each
universe $A$ the set of co-clones over $A$ then forms a
lattice when ordered by set inclusion, and given a set of relations
$\Gamma$ we write $\cclone{\Gamma}$ for the smallest co-clone
over $A$ containing $\Gamma$. Similarly, closure under upp-definitions
can also be viewed as a closure operator, and we write $\cclone{\Gamma}_{\exists !}$ for the
smallest set of relations over $A$ containing $\Gamma$ and which is closed under
upp-definitions.
Using these notions the question of the expressive strength of
upp-definitions can be stated as: for which sets of relations
$\Gamma$ is it the case that $\cclone{\Gamma} =
\cclone{\Gamma}_{\exists !}$? The main advantage behind this viewpoint
is that a co-clone $\cclone{\Gamma}$ can be described as the set of
relations {\em invariant} under a set of operations $F$, $\inv(F)$,
such that the operations in  $F$ describe all permissible combinations of tuples
in relations from $\Gamma$. An operation $f \in F$ is also said to be
a {\em polymorphism} of $\Gamma$ and if we let $\pol(\Gamma)$ be the
set of polymorphisms of $\Gamma$ then $\pol(\Gamma)$ is called a {\em
  clone}. This relationship allows us to
characterise the cases that need to be considered by using known properties of
$\pol(\Gamma)$, which is sometimes simpler than working only on the
relational side. This strategy will prove to be particularly useful for
Boolean sets of relations since all Boolean clones and co-clones have
been determined~\cite{pos41}.

\paragraph*{Our Results}
Our main research question is to identify $\Gamma$ such that
$\cclone{\Delta}_{\exists !} = \cclone{\Gamma}$ for each $\Delta$ such
that $\cclone{\Delta} = \cclone{\Gamma}$. If this holds we say that
$\cclone{\Gamma}$ is {\em $\exists !$-covered}. 
The main difficulty for proving this is that it might not be possible to directly transform a pp-definition
into an equivalent upp-definition. To mitigate this we analyse 
relations in co-clones using {\em partial polymorphisms}, which allows us to
analyse their expressibility in a very nuanced way.
In Section~\ref{sec:constructions} we show
how partial polymorphisms can be leveraged to prove that a given co-clone is
$\exists !$-covered. Most notably, we prove
that $\cclone{\Gamma}$ is $\exists !$-covered if $\pol(\Gamma)$
consists only of projections of the form $\pi(x_1,
\ldots, x_i, \ldots, x_n) = x_i$, or of projections and constant operations. As a
consequence, $\Gamma$ pp-defines all relations over $A$ if and only if
$\Gamma$ upp-defines all relations over $A$. One way of interpreting
this result is that if $\Gamma$ is ``sufficiently expressive''
then pp-definitions can always be turned into
upp-definitions. 
However, there also exists $\exists !$-covered co-clones where the
reason is rather that $\Gamma$ is ``sufficiently weak''.
For
example, if $\Gamma$ is invariant under the affine operation $x - y + z
\,(\mathrm{mod} \, |A|)$, then existential quantification does not add any
expressive power over unique existential quantification, since any existentially quantified variable occurring in a
pp-definition can be expressed via a linear equation, and is therefore
uniquely determined by other arguments. 
In Section~\ref{sec:bool} we then turn to the
Boolean domain, and obtain a full classification of the
$\exists !$-covered co-clones. Based on the results in
Section~\ref{sec:constructions} it is reasonable to expect that the covering
property holds for sufficiently expressive languages and sufficiently
weak languages, but that there may exist cases in between where unique
quantification differs from existential quantification.
This is indeed true, and we prove that the Boolean co-clones corresponding to
non-positive Horn
clauses, implicative and positive clauses, and their dual
cases, are not $\exists!$-covered. 
Last, in
Section~\ref{sec:complexity} we demonstrate how the results from
Section~\ref{sec:power} can be used for obtaining complexity
classifications of computational problems. 
One example of a problem requiring parsimonious reductions is the {\em unique satisfiability
  problem} over a Boolean set of relations $\Gamma$ ($\USAT(\Gamma)$) and its multi-valued generalization the {\em
  unique constraint satisfaction problem} ($\UCSP(\Gamma)$), where the goal is
to determine if there exists a unique model of a given conjunctive
$\Gamma$-formula. The complexity of $\USAT(\Gamma)$ was settled by
Juban~\cite{juban99} for finite sets of relations $\Gamma$,
essentially using a large case analysis. Using the results from
Section~\ref{sec:bool} this complexity classification can instead be
proved in a succinct manner, and we are also able to extend
the classification to infinite $\Gamma$ and large classes of
non-Boolean $\Gamma$. This systematic approach is also advantageous for proving lower
bounds, and we relate the complexity of $\USAT(\Gamma)$ to the highly
influential {\em exponential-time hypothesis} (ETH)~\cite{impagliazzo2001}, by
showing that none of the intractable cases of $\USAT(\Gamma)$ admit subexponential
algorithms without violating the ETH. 

\paragraph*{Related Work}
Primitive positive definitions with uniqueness quantification 
appeared in Creignou \& Hermann~\cite{CREIGNOU19961} in the context of
``quasi-equivalent'' logical formulas, and in the textbook by Creignou
et al.~\cite{cks01} under the name of {\em faithful implementations}. Similarly, upp-definitions were
utilised by Kavvadias \& Sideri~\cite{kavvadias98} to study the
complexity of the {\em inverse satisfiability problem}. 
A related topic is {\em frozen quantification}, which can be
viewed as uniqueness quantification restricted to variables that are constant
in any model~\cite{nordh2009}. 

\section{Preliminaries}
\label{sec:prel}

\subsection{Operations and Relations}

In the sequel, let $D \subseteq \mathbb{N}$ be a finite domain of
values. A $k$-ary function $f \colon D^{k} \rightarrow D$ is sometimes
referred to as an {\em operation} over $D$ and we write $\ar(f)$ to
denote the arity $k$. Similarly, a {\em partial operation} over $D$ is
a map $f \colon \domain(f) \rightarrow
D$ where $\domain(f) \subseteq D^k$ is called the {\em domain} of $f$,
and we let $\ar(f) = k$ be the arity of $f$. If $f$ and $g$ are
$k$-ary partial operations such that $\domain(f) \subseteq
\domain(g)$ and $f(t) = g(t)$ for each $t \in \domain(f)$ then $f$ is
said to be a {\em suboperation} of $g$.
For $k \geq 1$ and $1 \leq i \leq k$ we let
$\pi^{k}_i$ be the $i$th {\em projection}, i.e., $\pi^k_i(x_1, \ldots,
x_i, \ldots, x_{k}) = x_i$ for all $x_1, \ldots, x_i, \ldots, x_k \in
D$. We write $\functions{D}$ for the set of all operations over $D$ and
$\pfunctions{D}$ for the set of all partial operations over $D$.
As a notational shorthand we for $k \geq 1$ write $[k]$ for the set $\{1,
\ldots, k\}$. For $d \in D$ we by $\mathbf{d}^n$ denote the
constant $n$-ary tuple $(d, \ldots, d)$. Say that a $k$-ary $f \in
\functions{D}$ is {\em essentially unary} if there exists unary $g \in
\functions{D}$ and $i$ such that $f(x_1, \ldots, x_i, \ldots, x_n) =
g(x_i)$ for all $x_1, \ldots, x_i, \ldots, x_n \in D$.

Given an $n$-ary relation $R \subseteq D^n$ we write
$\ar(R)$ to denote its arity $n$. If $t = (x_1, \ldots, x_n)$ is an
$n$-ary tuple we write $t[i]$ to denote the
$i$th element $x_i$, and $\pro_{i_1, \ldots, i_{n'}}(t) = (t[i_1],
\ldots, t[i_{n'}])$ to denote the {\em
  projection} on the coordinates $i_1, \ldots,
i_{n'} \in \{1, \ldots, n\}$. Similarly, if $R$ is an $n$-ary relation 
we let $\pro_{i_1, \ldots, i_{n'}}(R) = \{\pro_{i_1, \ldots, i_{n'}}(t)
\mid t \in R\}$.
The $i$th argument of a relation $R$
is said to be {\em redundant} if there exists $j \neq i$ such that
$t[i] = t[j]$ for each $t \in R$, and is said to be {\em fictitious}
if for all $t \in R$ and $d \in D$ have $t' \in R$
where $t'[i] = d$ and $\pro_{1, \ldots, i - 1, i+1, \ldots,
  n}(t) = \pro_{1, \ldots, i - 1, i+1, \ldots,
  n}(t')$.

We write $\eq_D$ for the equality
relation $\{(x,x) \mid x \in D\}$ over $D$.  We will often represent
relations by their defining first-order formulas, and if $\varphi(x_1,
\ldots, x_n)$ is a first-order formula with $n$ free variables we
write $R(x_1, \ldots, x_n) \equiv \varphi(x_1, \ldots, x_n)$ to define
the relation $R = \{(f(x_1), \ldots, f(x_n)) \mid f$ is a model of
$\varphi(x_1, \ldots, x_n)\}$. We let $\relations{D}^n$ be the set of
all $n$-ary relations over $D$, $\relations{D}^{\leq n} = \bigcup^{n}_{i = 1} \relations{D}^{n}$, and $\relations{D} =
\bigcup^{\infty}_{i = 1} \relations{D}^i$. A set $\Gamma \subseteq
\relations{D}$ will sometimes be called a {\em constraint language}.
Each $k$-ary operation $f \in \functions{D}$ can be associated with a $(k+1)$-ary
relation $f^{\bullet} = \{(x_1, \ldots, x_k, f(x_1, \ldots, x_k)) \mid
x_1, \ldots, x_k \in D\}$, called the {\em graph} of $f$.

\subsection{Primitive Positive Definitions and Determined Variables}

We say that an $n$-ary relation $R$ has a {\em primitive positive definition}
(pp-definition) over a set of relations $\Gamma$ over a domain $D$ if $R(x_1, \ldots,
x_n) \equiv \exists y_1, \ldots, y_{n'} \colon R_1(\mathbf{x}_1) \land \ldots
\land R_m(\mathbf{x}_m)$ where each $\mathbf{x}_i$ is a tuple of
variables over $x_1, \ldots, x_n, y_1, \ldots, y_{n'}$ of length
$\ar(R_i)$ and each $R_i \in \Gamma \cup \{\eq_D\}$. Hence, $R$ can be
defined as a (potentially) existentially quantified conjunctive
formula over $\Gamma \cup \{\eq_D\}$. We will occasionally be
interested in pp-definitions not making use of existential
quantification, and call pp-definitions of this restricted type {\em
  quantifier-free primitive positive definitions} (qfpp-definitions).


\begin{definition}
  Let $R$ be an $n$-ary relation over a domain $D$. We say that $1 \leq i \leq n$ is
  {\em uniquely determined}, or just {\em determined}, if there exists
  $i_1,
  \ldots, i_k \in [n]$ and a function $h : D^{k} \rightarrow D$ such that $h(t[i_1], \ldots,
  \ldots, t[i_k]) = t[i]$ for each
  $t \in R$.
\end{definition}

When defining relations in terms of logical formulas we will occasionally
also say that the $i$th variable is uniquely determined, rather than
the $i$th index. 

\begin{definition} \label{def:uppdef2}
  An $n$-ary relation $R$ has a {\em unique primitive positive
    definition} (upp-definition) over a set of relations $\Gamma$ if
  there exists a pp-definition  \[R(x_1, \ldots, x_n) \equiv \exists y_1, \ldots, y_{n'} \colon R_1(\mathbf{x}_1) \land \ldots
\land R_m(\mathbf{x}_m)\] of $R$ over $\Gamma$  where each $y_i$ is uniquely determined by
$x_1, \ldots, x_n$.
\end{definition}

We typically write $\exists ! y_1, \ldots, y_{n'}$ for
the existentially quantified variables in a upp-definition. 
Following Nordh \& Zanuttini~\cite{nordh2009} we refer to unique
existential quantification over constant arguments as {\em frozen
existential quantification} ($i \in [\ar(R)]$ is constant if there
exists $d \in D$ such that $t[i] = d$ for each $t \in R$). If $R$ is
upp-definable over $\Gamma$ via a upp-definition only making use of
frozen existential quantification then we say that $R$ is {\em
freezingly pp-definable} (fpp-definable) over $\Gamma$. Let us define
the following closure operators over relations.

\begin{definition}
  Let $\Gamma$ be a set of relations. Then we define
  (1) $\cclone{\Gamma} = \{R \mid R$ has a pp-definition over $\Gamma\}$,
  (2), $\cclone{\Gamma}_{\exists !} = \{R \mid R$ has a upp-definition over
  $\Gamma\}$, 
  (3), $\cclone{\Gamma}_{\mathrm{fr}} = \{R \mid R$ has an fpp-definition over
  $\Gamma\}$, and
(4), $\pcclone{\Gamma}= \{R \mid R$ has a qfpp-definition over
  $\Gamma\}$.
\end{definition}

In all cases $\Gamma$ is called a {\em base}. If $\Gamma = \{R\}$ is singleton then we write $\cclone{R}$ instead of
$\cclone{\Gamma}$, and similarly for the other operators. Sets of relations of the form $\cclone{\Gamma}$ are usually called
{\em relational clones}, or {\em co-clones}, sets of the form
$\pcclone{\Gamma}$ {\em weak systems}, or {\em weak partial
  co-clones}, and sets of the form $\cclone{\Gamma}_{\mathrm{fr}}$ are
known as {\em frozen partial co-clones}. 
Note that $\cclone{\Gamma} \supseteq
\cclone{\Gamma}_{\exists !} \supseteq \cclone{\Gamma}_{\mathrm{fr}}
\supseteq \pcclone{\Gamma}$ for any $\Gamma \subseteq \relations{D}$.

Co-clones and weak systems can be described via algebraic invariants
known as {\em polymorphisms} and {\em partial polymorphism}. More
precisely, if $R \in \relations{D}^n$ and $f \in \functions{D}$ is a
$k$-ary operation, then for $t_1, \ldots, t_k \in R$ we let $f(t_1,
\ldots, t_k) = (f(t_1[1], \ldots, t_k[1]), \ldots, f(t_1[n], \ldots,
t_k[n]))$. We then say that a $k$-ary partial operation $f$ {\em
  preserves} an $n$-ary relation $R$ if $f(t_1, \ldots, t_k) \in R$ or
there exists $i \in [n]$ such that $(t_1[i], \ldots, t_k[i]) \notin
\domain(f)$, for each sequence of tuples $t_1, \ldots, t_k \in R$. If
$f$ preserves $R$ then $R$ is also said to be {\em invariant} under $f$. Note that if $f$ is total then the condition is
simply that $f(t_1, \ldots, t_k) \in R$ for each sequence $t_1,
\ldots, t_k \in R$. We then let $\ppol(R) = \{f \in \pfunctions{D} \mid f \textrm{ preserves }
  R\}$, $\pol(R) = \ppol(R) \cap \functions{D}$, $\ppol(\Gamma) =
  \bigcap_{R \in \Gamma}\ppol(R)$, and $\pol(\Gamma) = \bigcap_{R \in
    \Gamma}\pol(R)$. Similarly, if $F$ is a set of (partial)
  operations we let $\inv(F)$ be the set of relations
invariant under $F$, and write $\inv(f)$ if $F = \{f\}$ is singleton.
It is then known that $\inv(F)$ is a co-clone if $F \subseteq
\functions{D}$ and that $\inv(F)$ is a weak system if $F \subseteq
\pfunctions{D}$. More generally, $\cclone{\Gamma} =
\inv(\pol(\Gamma))$ and $\pcclone{\Gamma} =
\inv(\ppol(\Gamma))$, resulting in the following {\em Galois connections}.

\begin{theorem}[\cite{BKKR69i,BKKR69ii,Gei68,romov1981}]
  \label{theorem:galois}
  Let $\Gamma$ and $\Delta$ be two sets of relations. Then 
   $\Gamma \subseteq \cclone{\Delta}$ if and only if  $\pol(\Delta) \subseteq
    \pol(\Gamma)$ and $\Gamma \subseteq \pcclone{\Delta}$ if and only if  $\ppol(\Delta) \subseteq
    \ppol(\Gamma)$.
\end{theorem}


Last, we remark that sets of the form $\pol(\Gamma)$ and $\ppol(\Gamma)$ are usually called
{\em clones}, and {\em strong partial clones}, respectively, and form
lattices when ordered by set inclusion. Boolean clones are
particularly well understood and the induced lattice is known as {\em
  Post's lattice}~\cite{pos41}. If $F
\subseteq \functions{D}$ then we write
$\clone{F}$ for the
intersection of all clones over $D$ containing
$F$. Hence, $\clone{F}$ is the smallest clone over $D$ containing $F$.

\subsection{Weak and Plain Bases of Co-Clones}
\label{sec:prel_weak}
In this section we introduce two special types of bases of a co-clone,
that are useful for understanding the expressibility of
upp-definitions.

\begin{definition}[Schnoor \& Schnoor~\cite{schnoor2008a}] \label{def:weak}
  Let $\cclone{\Gamma}$ be a co-clone. A base $\Gamma_w$ of $\cclone{\Gamma}$ with the
  property that $\pcclone{\Gamma_w} \subseteq \pcclone{\Delta}$ for
  every base $\Delta$ of $\cclone{\Gamma}$ is called a {\em weak base}
  of $\cclone{\Gamma}$.
\end{definition}

Although not immediate from Definition~\ref{def:weak}, Schnoor \&
Schnoor~\cite{schnoor2008a} proved that a weak base
exists whenever $\cclone{\Gamma}$ admits a finite base, by 
the following relational construction.

\begin{definition} \label{def:core_construction}
    For $s \geq 1$ we let $\cols^s_D = \{t_1, \ldots, t_s\}$ where $t_1, \ldots, t_s$ is the sequence of $|D|^s$-ary tuples
    such that $(t_1[1], \ldots,
    t_s[1]), \ldots, (t_1[|D|^s], \ldots, t_s[|D|^s])$ is a
    lexicographic enumeration of $D^s$. 
\end{definition}
    Given a relation $R$ and a set of operations $F$ over a domain
    $D$, we let \[F(R)= \bigcap_{R' \in \inv(F), R \subseteq R' \in \relations{D}} R'.\]
We typically write $\cols^s$ instead
of $\cols^s_D$ if the domain $D$ is clear from the context, and
say that a co-clone $\inv(\cloneFont{C})$ has {\em core-size} $s$ if
there exist relations $R, R'$ such that $\pol(R) = \cloneFont{C}$, $R =
\cloneFont{C}(R')$, and $s = |R'|$. Weak bases can then be described
via core-sizes as follows (a clone $\cloneFont{C}$ is finitely related
if there exists a finite base of $\inv(\cloneFont{C})$).

\begin{theorem}[Schnoor \& Schnoor~\cite{schnoor2008a}] \label{thm:schnoor}
  Let $\cloneFont{C}$ be a finitely related clone where
  $\inv(\cloneFont{C})$ has core-size $s$. Then $\cloneFont{C}(\cols^t)$
  is a weak base of $\inv(\cloneFont{C})$ for every $t \geq s$.
\end{theorem}


\begin{table}
\centering
\caption{Relations.}
\scriptsize
\label{table:relations}
\begin{tabular}{l l}
  \hline
  Relation & Definition \\
  \hline
  {\scriptsize $\F$} &  {\scriptsize $\{(0)\}$} \\
  {\scriptsize $\T$} &  {\scriptsize $\{(1)\}$} \\
  {\scriptsize $\neqq$} &  {\scriptsize $\{(0,1), (1,0)\}$} \\       
 {\scriptsize $\evenn{}{n}$} & {\scriptsize $\{(x_1, \ldots, x_n) \in \{0,1\}^n \mid x_1 + \ldots
 + x_n$ is even$\}$} \\
 {\scriptsize $\evenn{n \neq}{n}$}  & {\scriptsize $\evenn{}{n}(x_1, \ldots, x_n) \land
 \neqq(x_1, x_{n+1}) \land \ldots \land \neqq(x_n, x_{2n})$ }\\
 {\scriptsize $\oddn{}{n}$} & {\scriptsize $\{(x_1, \ldots, x_n) \in \{0,1\}^n \mid x_1 + \ldots + x_n$ is odd$\}$} \\
 {\scriptsize $\oddn{n \neq}{n}$}  & {\scriptsize $\oddn{}{n}(x_1, \ldots, x_n) \land
 \neqq(x_1, x_{n+1}) \land \ldots \land \neqq(x_n, x_{2n})$}\\
  {\scriptsize $\nandn{}{n}$}  & {\scriptsize $\{0,1\}^n \setminus \{(1, \ldots, 1)\}$
  }\\
\hline
\end{tabular}                         
\end{table}

See Table~\ref{table:weak_plain_bases} for a list of
weak bases for the Boolean
co-clones of interest in this paper~\cite{Lagerkvist2014,lagerkvist2016}. Here, and in the
sequel, we use the co-clone terminology developed by Reith \& Wagner~\cite{reith2000}
and B\"ohler et al.~\cite{bcrv03},
where a Boolean co-clone $\inv(\cloneFont{C})$ is typically written as $\cc{IC}{}{}$.
Many relations in Table~\ref{table:weak_plain_bases}
are provided by their defining logical formulas; for example, $x_1
\rightarrow x_2$ is the binary relation $\{(0,0), (0,1), (1,1)\}$. See
Table~\ref{table:relations} for definitions of the remaining
relations. As a convention we use $c_0$ to indicate a variable which
is constant 0 in any model, and $c_1$ for a variable which is constant 1.
On the functional side we use the bases by B\"ohler et al.~\cite{bcrv03} and let $\cc{I}{}{2} = \clone{\pi^1_1}, \cc{I}{}{0} = \clone{0},
\cc{I}{}{1} = \clone{1}, \cc{I}{}{} = \clone{\{0,1\}}$, $\cc{N}{}{2} =
\clone{\overbar{x}}$, $\cc{N}{}{} = \clone{\{\overbar{x}, 0, 1\}}$,
$\cc{E}{}{2} = \clone{\land}$, $\cc{E}{}{0} = \clone{\{\land, 0\}}$,
$\cc{E}{}{1} = \clone{\{\land, 1\}}$, $\cc{E}{}{} = \clone{\{\land, 0,
  1\}}$, $\cc{L}{}{2} =
  \clone{x \oplus y \oplus z}$, and $\cc{S}{}{11} = \clone{\{x \land (y \lor z), 0\}}$, where
$\overbar{x} = 1 - x$ and where $0,1$ are shorthands for the two
constant Boolean operations.
We conclude this section by defining the dual notion of a weak base.

\begin{definition}[Creignou et al.~\cite{creignou2008}]
  Let $\cclone{\Gamma}$ be a co-clone. A base $\Gamma_p$ of $\cclone{\Gamma}$ with the
  property that $\pcclone{\Delta} \subseteq \pcclone{\Gamma_p}$ for
  every base $\Delta$ of $\cclone{\Gamma}$ is called a {\em plain base}
  of $\cclone{\Gamma}$.
\end{definition}

Clearly, every co-clone is a trivial plain base of itself,
but the question remains for which co-clones more succinct plain bases
can be found. For arbitrary finite domains little is known but in the
Boolean domain succinct plain bases have been described~\cite{creignou2008} (see Table~\ref{table:weak_plain_bases}).



\begin{table}  
\centering
\caption{Weak and plain bases of selected Boolean co-clones.}
\tiny
\label{table:weak_plain_bases}
\begin{tabularx}{\textwidth}{l l l}
  \hline
  $\cloneFont{C}$ & Weak base of $\inv(\cloneFont{C})$ & Plain base of $\inv(\cloneFont{C})$\\
  \hline

  {\scriptsize $\cc{S}{n}{1}$} &  {\scriptsize $\{\nandn{}{n}(x_1, \ldots, x_n) \wedge
  \F(c_0)\}$} & {\scriptsize $\{\nandn{}{n}\}$}  \\

  {\scriptsize $\cc{S}{}{1}$} & {\scriptsize $\{\nandn{}{\ra}(x_1, \ldots, x_{\ra}) \wedge
  \F(c_0) \mid \ra \geq 2\}$} & {\scriptsize $\{\nandn{}{\ra} \mid \ra \geq 1\}$}  \\

  {\scriptsize $\cc{S}{n}{12}$} &  {\scriptsize $\{\nandn{}{n}(x_1, \ldots, x_n) \wedge
  \F(c_0) \wedge \T(c_1)\}$} & {\scriptsize $\{\nandn{}{n}, \T(c_1)\}$}  \\

  {\scriptsize $\cc{S}{}{12}$} & {\scriptsize $\{\nandn{}{\ra}(x_1, \ldots, x_{\ra}) \wedge \F(c_0) \wedge
      \T(c_1) \mid \ra \geq 2\}$} & {\scriptsize $\{\nandn{}{\ra} \mid \ra \geq 1 \} \cup  \{\T(c_1)\}$} \\

  {\scriptsize $\cc{S}{n}{11}$} & {\scriptsize $\{\nandn{}{n}(x_1, \ldots, x_n)
  \wedge (\neg x \rightarrow \neg x_1 \cdots \neg x_n)
   \wedge \F(c_0)\}$} & {\scriptsize $\{\nandn{}{n}, (x_1 \rightarrow x_2)\}$}  \\

   {\scriptsize $\cc{S}{}{11}$} & {\scriptsize $\{R_{\cc{S}{\ra}{11}} \mid \ra \geq 2\}$} & {\scriptsize $\{\nandn{}{\ra} \mid \ra \geq 1\} \cup  \{(x_1 \rightarrow x_2)\}$} \\

  {\scriptsize $\cc{S}{n}{10}$} & {\scriptsize
    $\{R_{\cc{S}{n}{11}}(x_1, \ldots, x_{\ra}, c_0) \land \T(c_1)\}$}
  & {\scriptsize $\{\nandn{}{n}, (x_1 \rightarrow x_2), \T(c_1)\}$} \\

   {\scriptsize $\cc{S}{}{10}$} & {\scriptsize $\{R_{\cc{S}{n}{10}} \mid \ra \geq 2\}$} & {\scriptsize $\{\nandn{}{\ra} \mid \ra \geq 1\}  \cup \{(x_1 \rightarrow x_2), \T(c_1)\}$} \\

  {\scriptsize $\cc{D}{}{}$} & {\scriptsize  $\{(x_1 \oplus x_2 =
    1)\}$} & {\scriptsize $\{(x_1 \oplus x_2 = 1)\}$} \\
  {\scriptsize $\cc{D}{}{1}$} & {\scriptsize $\{(x_1 \oplus x_2 = 1)
    \wedge \F(c_0)\} \wedge \T(c_1)$} & {\scriptsize $\{(x_1 \oplus x_2 = 1)\}
  \cup \{\F(c_0), \T(c_1)\}$}  \\
  {\scriptsize $\cc{D}{}{2}$} & {\scriptsize $\{(x_1 \lor x_2) \land \neqq(x_1, x_3)
  \land \neqq(x_2, x_4) \wedge \F(c_0) \wedge \T(c_1)\}$} &
  {\scriptsize $\{(x_1 \vee x_2), (\neg
  x_1 \vee x_2), (\neg x_1 \vee \neg x_2)\}$}  \\

  {\scriptsize $\cc{E}{}{}$} & {\scriptsize $\{(x_1 \leftrightarrow x_2x_3) \wedge (x_2 \vee x_3
\rightarrow x_4)\}$} & {\scriptsize $\{(\neg x_1 \vee \ldots \vee \neg x_k \vee x) \mid k
  \geq 1\}$} \\
  {\scriptsize $\cc{E}{}{0}$} & {\scriptsize $\{(x_1 \leftrightarrow x_2x_3) \wedge (x_2 \vee x_3
\rightarrow x_4) \wedge \F(c_0)\}$} & {\scriptsize $\{\nandn{}{n} \mid n
  \in \mathbb{N}\} \cup \{(\neg x_1 \vee \ldots \vee \neg x_k \vee x)
  \mid k \geq 1\}$}  \\
  {\scriptsize $\cc{E}{}{1}$} & {\scriptsize $\{(x_1 \leftrightarrow
    x_2x_3) \wedge \T(c_1)\}$} & {\scriptsize $\{(\neg x_1 \vee \ldots \vee \neg x_k \vee x) \mid k
  \in \mathbb{N}\}$}  \\
  {\scriptsize $\cc{E}{}{2}$} & {\scriptsize $\{(x_1 \leftrightarrow
    x_2x_3) \wedge \F(c_0) \wedge \T(c_1)\}$} & {\scriptsize $\{\nandn{}{n} \mid n
  \in \mathbb{N}\} \cup \{(\neg x_1 \vee \ldots \vee \neg x_k \vee x) \mid k \in \mathbb{N}\}$}  \\


  \hline

\end{tabularx}
\end{table}

\subsection{Duality}
Many questions concerning Boolean co-clones can be simplified by only
considering parts of Post's lattice. If $f \in \functions{\{0,1\}}$ is
$k$-ary then the {\em dual} of $f$, $\mathrm{dual}(f)$, is the
operation $\mathrm{dual}(f)(x_1, \ldots, x_k) =
\overbar{f(\overbar{x_1}, \ldots, \overbar{x_k})}$, and we let
$\mathrm{dual}(F) = \{\mathrm{dual}(f) \mid f \in F\}$ for a set $F
\subseteq \functions{\{0,1\}}$. Each Boolean clone $\cloneFont{C}$ can
then be associated with a dual clone
$\mathrm{dual}(\cloneFont{C})$. Similarly, for $R \in
\relations{\{0,1\}}$ we let $\mathrm{dual}(R) = \{\overbar{t} \mid t
\in R\}$ and $\mathrm{dual}(\Gamma) = \{\mathrm{dual}(R) \mid R \in
\Gamma\}$ for $\Gamma \subseteq \relations{\{0,1\}}$. It is then known
that $\inv(\mathrm{dual}(\cloneFont{C})) = \mathrm{dual}(\inv(\cloneFont{C}))$.


\section{The Expressive Power of Unique Existential Quantification} 
\label{sec:power}

The main goal of this paper is to understand when the expressive power
of unique existential quantification coincides with existential
quantification in primitive positive formulas. Let us first consider
an example where a pp-definition can be rewritten into a upp-definition.

\begin{example} \label{ex:upp}
  Consider the canonical reduction from $k$-SAT to $(k-1)$-SAT via
  pp-definitions of the form $(x_1 \lor \ldots \lor x_k) \equiv
  \exists y \colon (x_1 \lor \ldots \lor x_{k-2} \vee y) \land (x_{k-1} \lor
  x_k \lor \neg y)$. In this pp-definition the auxiliary variable $y$ is {\em not}
  uniquely determined since, for example, $y = 0$ and $y = 1$ are both
  consistent with 
  $x_1 = 1, \ldots x_{k-2} = 1, x_{k-1} = 1, x_k = 1$. On the other hand, if we instead
  take the pp-definition $(x_1 \lor \ldots \lor x_k) \equiv
  \exists y \colon (x_1 \lor \ldots \lor x_{k-2} \lor y) \land (y
  \leftrightarrow (x_{k-1} \vee x_k))$, which can be expressed by $(k
  - 1)$-SAT, it is easily verified that $y$ is determined by $x_{k-1}$
  and $x_k$.
\end{example}

Using the algebraic terminology from Section~\ref{sec:prel} this 
property can be phrased as follows.

\begin{figure} 
  \centering
  \includegraphics[scale=0.8]{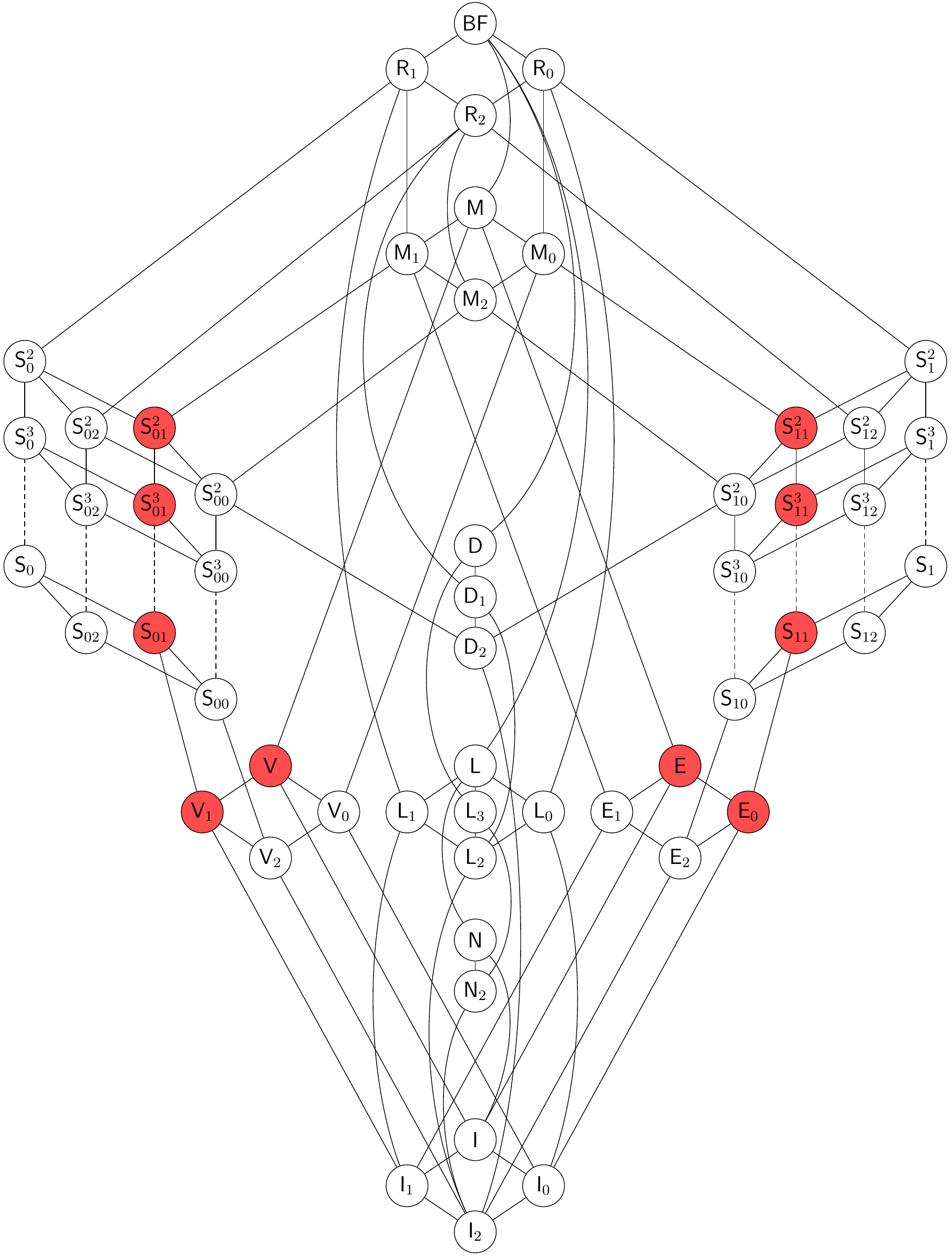}
  \caption{The lattice of Boolean clones. $\inv(\cloneFont{C})$
    is coloured in red if and only if $\inv(\cloneFont{C})$
  is not $\exists !$-covered.}
\label{fig:lattice}
\end{figure}

\begin{definition} \label{def:cover}
A co-clone $\cclone{\Gamma}$ is {\em $\exists !$-covered} if
$\cclone{\Gamma} = \cclone{\Delta}_{\exists !}$ for every base $\Delta$ of $\cclone{\Gamma}$.
\end{definition}

Thus, we are interested in determining the $\exists !$-covered co-clones, and since every constraint language $\Gamma$
belongs to a co-clone, namely $\cclone{\Gamma}$,
Definition~\ref{def:cover} precisely captures the aforementioned
question concerning the expressive strength of uniqueness
quantification in primitive positive formulas. The remainder of this
section will be dedicated to proving covering results of this form, with a
particular focus on proving a full classification for the Boolean
domain. See Figure~\ref{fig:lattice} for a visualisation of 
this dichotomy.
We begin in Section~\ref{sec:constructions}
by outlining some of the main ideas required to prove that a co-clone
is $\exists !$-covered, and consider covering results
applicable for arbitrary finite domains. In Section~\ref{sec:bool} we
turn to the Boolean domain where we prove the classification in
Figure~\ref{fig:lattice}. 

\subsection{General Constructions}
\label{sec:constructions}

Given an arbitrary constraint language $\Gamma$ it can be difficult to
directly reason about the strength of upp-definitions over $\Gamma$. Fortunately, there are methods to mitigate this
difficulty. Recall from Definition~\ref{def:weak} that a weak base of
a co-clone $\cclone{\Gamma}$ is a base which is qfpp-definable by any
other base of $\cclone{\Gamma}$, and that a plain base is a base with
the property that it can qfpp-define every relation in the
co-clone. We then have the following useful lemma.

\begin{lemma} \label{lemma:weak_upp}
  Let $\cclone{\Gamma}$ be a co-clone with a weak base $\Gamma_w$ and a plain
  base $\Gamma_p$. If $\Gamma_p \subseteq \cclone{\Gamma_w}_{\exists !}$ then $\cclone{\Gamma}$ is $\exists !$-covered.
\end{lemma}

\begin{proof}
  Let $\Delta$ be a base of $\cclone{\Gamma}$ and take an arbitrary $n$-ary
  relation $R \in \cclone{\Gamma}$. First, take a qfpp-definition
  $R(x_1, \ldots, x_n) \equiv \varphi(x_1, \ldots, x_n)$ over
  $\Gamma_p$. By assumption, $\Gamma_w$ can upp-define every relation
  in $\Gamma_p$, and it follows that \[R(x_1, \ldots, x_n) \equiv
  \exists ! y_1, \ldots, y_m \colon \varphi'(x_1, \ldots, x_n, y_1,
  \ldots, y_m)\] for a $\Gamma_w$-formula $\varphi'(x_1, \ldots, x_n, y_1, \ldots, y_m)$
  since each constraint in $\varphi(x_1, \ldots,
  x_n)$ can be replaced by its upp-definition over $\Gamma_w$. Last, since $\Delta$
  can qfpp-define $\Gamma_w$, we can obtain a upp-definition of $R$ by
  replacing each constraint in $\varphi'(x_1, \ldots, x_n, y_1,
  \ldots, y_m)$ by its qfpp-definition over $\Delta$. 
\end{proof}

Although not difficult to prove, Lemma~\ref{lemma:weak_upp} offers the
advantage that it is sufficient to prove that $\Gamma_p \subseteq \cclone{\Gamma_w}_{\exists !}$ for two constraint languages
$\Gamma_w$ and $\Gamma_p$. 
Let us now illustrate some additional techniques for proving that
$\cclone{\Gamma}$ is $\exists !$-covered. Theorem~\ref{thm:schnoor} in Section~\ref{sec:prel_weak} shows that
the relation $\cloneFont{C}(\cols^s)$ is a weak base of
$\inv(\cloneFont{C})$ for $s$ larger than or equal to the core-size of
$\inv(\cloneFont{C})$.  For $s$ smaller than the core-size we have the
following description of $\cloneFont{C}(\cols^s)$.

\begin{theorem} \label{thm:core_size}
  Let $\cloneFont{C}$ be a finitely related clone over a finite domain $D$. Then, for
  every $s \geq 1$,
  $\cloneFont{C}(\cols^s) \in \pcclone{\Gamma}$ for every 
  base $\Gamma$ of $\inv(\cloneFont{C})$.
\end{theorem}

  \begin{proof}
The intuitive meaning behind the relation $\cloneFont{C}(\cols^s)$ is
that it may be
viewed as a relational representation of the set of all $s$-ary
operations of a clone $\cloneFont{C}$, in the sense that there for each
$s$-ary $f \in \cloneFont{C}$ exists $t_f \in \cloneFont{C}(\cols^s)$
such that $f(t_1, \ldots, t_s) = t_f$, where $\{t_1, \ldots, t_s\} =
\cols^s$. 
Moreover, the operation $f$ preserves  $R \in \Gamma$ if
    and only if $(f(t_1[1], \ldots, t_s[1]), \ldots, f(t_1[n], \ldots,
    t_s[n])) \in R$. In a qfpp-definition of $\cloneFont{C}(\cols^s)$ we
    then associate each variable $x$ with an element of $D^s$, and then
    for each $R \in \Gamma$ and $t_1, \ldots, t_s \in R$ add the
    constraint $R((t_1[1], \ldots, t_s[1]), \ldots, (t_1[n], \ldots,
    t_s[n]))$. For further details, see Theorem 2 in Bodnarchuk et
    al.~\cite{BKKR69ii}, or  Theorem 15 in Dalmau~\cite{dalmau2000}.
\end{proof}

The applications of Theorem~\ref{thm:core_size} in the context of
upp-definitions might not be immediate. However, observe that each
argument $i \in [|D|^s]$ of $\cols^s$ is determined by at most $s$
other arguments, and if $\cloneFont{C}$ is sufficiently simple, this
property can be proved to hold also for $\cloneFont{C}(\cols^s)$. This
intuition can then be formalised into the following general theorem.

\begin{theorem} \label{thm:ess}
  Let $\pol(\Gamma)$ be a clone over a finite domain $D$ such that
  each $f \in \pol(\Gamma)$ is a constant
  operation or a projection. Then
  $\cclone{\Gamma}$ is $\exists !$-covered.
\end{theorem}
\begin{proof}
  Let $F$ be a set of operations such that $\clone{F} = \pol(\Gamma)$. We may without loss of
  generality assume that $F = \{f_1, \ldots, f_k\}$ for unary
  operations $f_l$ such that $f_l(x) = d_l$ for some $d_l \in D$.  
  Take an arbitrary $n$-ary relation $R \in
  \cclone{\Gamma}$. Let $s = |R|$ and consider the relation
  $F(\cols^s)$ from Definition~\ref{def:core_construction}. Our aim is to prove
  that $F(\cols^s)$ can upp-define $R$, which is sufficient since
  $F(\cols^s) \in \pcclone{\Gamma}$ via Theorem~\ref{thm:core_size}. Let $i_1, \ldots, i_n \in
  [|D|^s]$ denote the indices satisfying $\pro_{i_1, \ldots,
    i_n}(F(\cols^s)) = R$. 


  If $k = 0$, and $\pol(\Gamma)$ consists only of projections, then
  $F(\cols^s) = \cols^s$, and each argument in $[|D|^s] \setminus
  \{i_1, \ldots, i_n\}$ is already determined by $i_1, \ldots,
  i_n$, and by the preceding remark $R \in
  \cclone{F(\cols^s)}_{\exists !}$. Therefore, assume that $k \geq 1$. For each $f_l \in F$ then observe that $(d_l, \ldots,
  d_l) \in F(\cols^s)$ and that $(d_l, \ldots, d_l) \in \pro_{i_1,
    \ldots, i_n}(\cols^s)$. Choose $j_1, j_2 \in [|D|^s]$ such that  $t[j_1] \neq
  t[j_2]$ for $t \in \cols^s$ if and only if $\pro_{i_1, \ldots, i_n}(t)
  = (d_l, \ldots, d_l)$, for a $d_l$ such that $f_l(x) = d_l$. Thus, we
  choose a pair of indices differing in $\cols^s$ if and only if the
  projection on $i_1, \ldots, i_n$ is constant.  Such a
  choice is always possible since the arguments of $\cols^s$ enumerate
  all $s$-ary tuples over $D$. Then construct
  the relation $R'(x_1, \ldots, x_{|D|^s}) \equiv F(\cols^s)(x_1,
  \ldots, x_{|D|^s}) \land \eq(x_{j_1}, x_{j_2})$. 
  It follows that $\pro_{i_1,
  \ldots, i_n}(R') = R$, and that every argument $l \in
[|D|^s] \setminus \{i_1, \ldots, i_n\}$ is determined by $i_1, \ldots,
i_n$. Hence, $R \in \cclone{F(\cols^s)}_{\exists !}$.
\end{proof}

Theorem~\ref{thm:ess} implies that $\cclone{\Gamma}$ is $\exists
!$-covered if $\Gamma$ is sufficiently powerful, and in particular
implies that $\relations{D}$ is $\exists !$-covered for every finite
$D$. Hence, $\Gamma$ pp-defines every relation if and only if $\Gamma$
upp-defines every relation. However, as we will now illustrate, this is not the only possible
case when a co-clone is $\exists !$-covered.

\begin{lemma} \label{lemma:all_determined}
  Let $F$ be a set of operations over a finite domain $D$. If each
  argument $i \in [\ar(R)]$ is either fictitious or determined for every $R \in \inv(F)$, then $\inv(F)$ is $\exists !$-covered.
\end{lemma}

\begin{proof}
  Let $\Gamma$ be a set of relations such that $\cclone{\Gamma} =
  \inv(F)$, and let $R \in \inv(F)$ be an $n$-ary relation. Let
  $R(x_1, \ldots, x_n) \equiv \exists y_1, \ldots, y_m \colon
  \varphi(x_1, \ldots, x_n, y_1, \ldots, y_m)$ denote a pp-definition
  of $R$ over $\Gamma$. First consider the relation $R^m(x_1, \ldots,
  x_n, y_1, \ldots, y_m) \equiv \varphi(x_1, \ldots, x_n, y_1, \ldots,
  y_m)$ which is upp-definable (indeed, even qfpp-definable) over
  $\Gamma$. Hence, $R^m$ is preserved by $F$, implying that the $(n+m)$th argument is
  either fictitious or determined. In the first case  we construct
  the relation 
  \begin{align*}
  R^{m-1}(x_1,
  \ldots, x_n, y_1, \ldots, y_{m-1}) \equiv \, & \exists ! y_{m} \colon
  \eq_D(x_1, y_m) \land \\ & \varphi(x_1, \ldots, x_n, y_1, \ldots,
  y_m).
  \end{align*}
 In the second case, we can directly upp-define the $(n+m-1)$-ary relation
  $R^{m-1}$ as 
  \begin{align*}
R^{m-1}(x_1,
  \ldots, x_n, y_1, \ldots, y_{m-1}) \equiv \, & \exists ! y_{m}
                                          \colon \\ & \varphi(x_1, \ldots, x_n, y_1, \ldots,
  y_m).
  \end{align*}
 Since $R^{m-1} \in \cclone{\Gamma}_{\exists !} \subseteq
  \inv(F)$, it is clear that this procedure can be repeated until the
  relation $R$ is upp-defined.
\end{proof}

\begin{theorem} \label{thm:affine}
  Let $D$ be a finite domain such that $|D|$ is prime, and let
  $f(x,y,z) = x - y +  z \, (\bmod\, |D|)$. Then, for any constraint
  language $\Gamma$ over $D$ such that $\cclone{\Gamma} \subseteq \inv(f)$, $\cclone{\Gamma}$ is
  $\exists !$-covered.
\end{theorem}

\begin{proof}
We will prove that the preconditions of
Lemma~\ref{lemma:all_determined} are satisfied for $\inv(f)$, which is sufficient to
prove the claim. Let $R$ be invariant under $f$. Then it is known that $R$ is the solution space
of a system of linear equations modulo $|D|$~\cite{jeavons1995}, from
which it follows that each argument is either determined, since it can
be written as a unique combination of other arguments, or is
fictitious. 
\end{proof}

\subsection{Boolean Constraint Languages}
\label{sec:bool}

In this section we use the techniques developed so far to prove that
the classification in Figure~\ref{fig:lattice} is correct. Note first
that $\inv(\cloneFont{C})$ is $\exists !$-covered if and only if
$\inv(\mathrm{dual}(\cloneFont{C}))$ is $\exists !$-covered, since a
upp-definition $\exists ! y_1, \ldots, y_{n'} \colon R_1(\mathbf{x}_1)
\land \ldots \land R_m(\mathbf{x}_m)$ of $n$-ary $R \in
\inv(\cloneFont{C})$ immediately yields a upp-definition $\exists ! y_1, \ldots, y_{n'} \colon \mathrm{dual}(R_1)(\mathbf{x}_1)
\land \ldots \land \mathrm{dual}(R_m)(\mathbf{x}_m)$ of
$\mathrm{dual}(R) \in \inv(\mathrm{dual}(\cloneFont{C}))$.
Thus, to simplify
the presentation we omit the case when 
$\cloneFont{C} \supseteq \cc{V}{}{2}$ in Figure~\ref{fig:lattice}.
Let us
begin with the cases following directly from
Section~\ref{sec:constructions} or from existing results (recall that
$\cc{IC}{}{}$ is a shorthand for $\inv(\cloneFont{C})$).

\begin{lemma} \label{lemma:simple}
  Let $\cc{IC}{}{}$ be a Boolean co-clone. Then $\cc{IC}{}{}$ is
  $\exists !$-covered if $\cc{IC}{}{} \subseteq \cc{IM}{}{2}$, $\cc{IC}{}{}
  \subseteq \cc{IL}{}{2}$, $\cc{IC}{}{} \subseteq \cc{IS}{}{12}$,
  $\cc{IC}{}{} = \cc{IS}{}{10}$, $\cc{IC}{}{} = \cc{IS}{n}{10}$ for
  some $n \geq 2$, $\cc{IC}{}{} = \cc{IS}{}{1}$, or $\cc{IC}{}{} =
  \cc{IS}{n}{1}$ for some $n \geq 2$.
\end{lemma}

\begin{proof}
  The case when $\cc{IC}{}{} \subseteq \cc{IL}{}{2}$ follows from
  Theorem~\ref{thm:affine} since $\cc{L}{}{2} =
  \clone{x \oplus y \oplus z}$. 
  For each case when 
  $\cc{C}{}{}$ belongs to the infinite chains in Post's lattice, or if
  $\cc{IC}{}{} \subseteq \cc{IM}{}{2}$, it is
  known that $\cc{IC}{}{} = \cclone{\Gamma}_{\mathrm{fr}}$ for any
  base $\Gamma$ of $\cc{IC}{}{}$~\cite{nordh2009}, which is sufficient
  since $\cclone{\Gamma}_{\mathrm{fr}} \subseteq
  \cclone{\Gamma}_{\exists !}$. 
\end{proof}

We now move on to the more interesting cases, and begin with the case
when $\pol(\Gamma)$ is essentially unary, i.e., consists of
essentially unary operations. This covers
$\cc{I}{}{2}, \cc{I}{}{0}, \cc{I}{}{1}, \cc{I}{}{},
\cc{N}{}{2},\cc{N}{}{}$ from Figure~\ref{fig:lattice}.

\begin{theorem} \label{thm:ess_boolean}
  Let $\Gamma$ be a Boolean constraint language such that
  $\pol(\Gamma)$ is essentially unary. Then $\cclone{\Gamma}$ is
  $\exists !$-covered.
\end{theorem}

\begin{proof}
  From Theorem~\ref{thm:ess} only the two co-clones
  $\cc{IN}{}{}$ and $\cc{IN}{}{2}$ remain, where $\cc{IN}{}{} =
  \inv(\{\overbar{x}, 0, 1\})$ and $\cc{IN}{}{2} =
  \inv(\overbar{x})$. The two cases are similar and we for brevity
  concentrate only on $\cc{IN}{}{2}$. Hence, let  $R \in \cc{IN}{}{2}$ be
  an $n$-ary relation, which we without loss of generality may assume
  has no redundant arguments. Since $\cc{N}{}{2} = \clone{\overbar{x}}$ we
  begin by partitioning $R$ into two disjoint sets $S$ and
  $\overbar{S}$ where $t \in S$ if and only if $\overbar{t} \in
  \overbar{S}$. Let $s = |S| = |\overbar{S}|$, and construct the relation
  $\{\overbar{x}\}(\cols^s)$, which is qfpp-definable over $\Gamma$
  according to Theorem~\ref{thm:core_size}. Let $i_1, \ldots, i_n \in
  \{1, \ldots, n\}$ be the set of indices satisfying $\pro_{i_1,
    \ldots, i_n}(\cols^s) = S$ and $\pro_{i_1,
    \ldots, i_n}(\{\overbar{x}\}(\cols^s) \setminus \cols^s) =
  \overbar{S}$. Assume there exists $i \in [2^{s}] \setminus \{i_1,
  \ldots, i_n\}$ such that $i$ is not determined by $i_1, \ldots,
  i_n$. By construction, $i$ is determined by $i_1, \ldots, i_n$ in
  both $\cols^s$ and $\{\overbar{x}\}(\cols^s) \setminus \cols^s$,
  implying that the only possible outcome is the existence of $t \in
  \cols^s$ and $t' \in \{\overbar{x}\}(\cols^s) \setminus \cols^s$
  where $t[i] \neq t'[i]$, but $\pro_{i_1, \ldots, i_n}(t) =
  \pro_{i_1, \ldots, i_n}(t')$. But then $\overbar{t'} \in \cols^s$,
  and since $\pro_{i_1, \ldots, i_n}(\cols^s) = S$ and $\pro_{i_1,
    \ldots, i_n}(t) = \pro_{i_1, \ldots, i_n}(t')$,
  $\overbar{\pro_{i_1, \ldots, i_n}(t')} \in S$, contradicting the
  partitioning of $R$ into the disjoint sets $S$ and
  $\overbar{S}$. Then, since every argument $i$ distinct from $i_1,
  \ldots, i_n$ is determined by $i_1, \ldots, i_n$, the original
  relation $R$ can be upp-defined by $\{\overbar{x}\}(\cols^s)$ 
  using only unique existential quantification.

  The $\cc{IN}{}{}$ case is similar, and the only difference is that
  we partition the input relation $R$ into $S, \overbar{S}$, and
  $\zeroes, \ones$, and proceed with $s = |S| = |\overbar{S}|$. 
\end{proof}

Next, we consider the co-clone $\cc{ID}{}{2}$, consisting of all
relations pp-definable by binary clauses.

\begin{lemma} \label{lemma:id2}
$\cc{ID}{}{2}$ is $\exists !$-covered.
\end{lemma}

\begin{proof}
  We will show the result by using Lemma~\ref{lemma:weak_upp}.
  According to Table~\ref{table:weak_plain_bases} the relation
$R_w(x_1, x_2, x_3, x_4, c_0, c_1) \equiv$ \[(x_1 \lor x_2) \land \neqq(x_1, x_3)
  \land \neqq(x_2, x_4) \land
   \F(c_0) \wedge \T(c_1)\] is a weak base of
$\cc{ID}{}{2}$, and the plain base $\Gamma_p$ is given by $\Gamma_p = \{(x_1 \vee x_2), (\neg
  x_1 \vee x_2), (\neg x_1 \vee \neg x_2)\}$. Hence, we need to prove
  that $R_w$ can upp-define each relation in $\Gamma_p$. Now consider the
  following upp-definitions:
\[(x_1 \lor x_2) \equiv \exists! x_3, x_4, c_0, c_1
\colon R_w(x_1, x_2, x_3, x_4, c_0, c_1),\]
\[(\neg x_1 \lor x_2) \equiv \exists! x_3, x_4, c_0, c_1
\colon R_w(x_3, x_2, x_1, x_4, c_0, c_1),\]
and
\[(\neg x_1 \lor \neg x_2) \equiv \exists! x_3, x_4, c_0, c_1
\colon R_w(x_3, x_4, x_1, x_2, c_0, c_1).\]

In each case it is readily verified that the existentially quantified
variables are indeed uniquely determined. For example, in the
upp-definition of $(\neg x_1 \lor x_2)$ the variable $x_3$ is uniquely
determined by $x_1$ since $t[1] \neq t[3]$ for every $t \in R_w$.
\end{proof}

We now tackle the cases when $\inv(\{\land, 0, 1\}) \subseteq
\cc{IC}{}{} \subseteq \inv(\{\land\})$, which in Figure~\ref{fig:lattice}
corresponds to $\cc{E}{}{}$, $\cc{E}{}{0}$, $\cc{E}{}{1}$, and
$\cc{E}{}{2}$. As a first step we begin by characterising the
determined arguments of relations in $\cc{E}{}{0}$.

\begin{lemma} \label{lemma:upp_ie0}
  Let $R \in \cc{IE}{}{0}$ be an $n$-ary relation. If $i \in [n]$ is determined in $R$ then 
    either (1) there exists $i_1, \ldots, i_k \in [n]$ distinct from $i$ such that $t[i] = t[i_1] \land
    \ldots \land t[i_k]$ for every $t \in R$, or
    (2) $t[i] = 0$ for every $t \in R$.
\end{lemma}

\begin{proof} 
  Assume that $i \in [n]$
  is determined in $R$. Let $R_1 = \{t_1, \ldots, t_m\} = \{t \in R \mid
  t[i] = 1\}$ and $R_0 = \{s_1, \ldots, s_{m'}\} = \{s \in R \mid
  t[i] = 0\}$. Note first that $R_0 = \emptyset$ cannot happen since
  $R$ is preserved by 0, and if $R_1 = \emptyset$ then we
  end up in case (2). Hence, in the remainder of the proof we assume
  that  $R_0$ and $R_1$ are both non-empty. 


    Consider the tuple $t_1 \land
    \ldots \land t_m = t$ (applied componentwise), and observe that $t
    \in \{t_1, \ldots, t_m\}$ since $R$ is preserved by $\land$, and that
    $t[i] = 1$ since $t_1[i] = \ldots = t_m[i] = 1$. 
    Furthermore, if $t[j] = 1$ for some $j \in [n]$ then it must also
    be the case that $t_1[j] = \ldots = t_m[j] = 1$. Let $i_1, \ldots, i_l \in [n] \setminus \{i\}$ denote
    the set of indices such that $t[i_j ]= 1$. Then $t'[i] = t'[i_1] \land
    \ldots \land t'[i_l]$ for every $t' \in R_1$, and we also
    claim that $s[i] = s[i_1] \land
    \ldots \land s[i_l]$ for every $s \in R_0$, thus ending up in case
    (1). Note that $l > 0$, as
    otherwise every argument distinct from $i$ is constantly 0 in
    $t$, which is not consistent with the fact that $\mathbf{0}^n \in R_0$, since
    it contradicts the assumption that $i$ is determined.
    Assume that there exists $s \in
    R_0$ such that $s[i]  = 0 \neq s[i_1] \land \ldots \land s[i_l]$. Then,
    clearly, $s[i_1] = \ldots = s[i_l] = 1$. But then $t \land s \in R$
    implies that $i$ is not determined, since $\pro_{1, \ldots, i-1,
      i+1, \ldots, n}(t \land s) = \pro_{1, \ldots, i-1, i+1, \ldots,
      n}(t)$ but $(t \land s)[i] \neq t[i]$. Hence, $s[i] = s[i_1] \land
    \ldots \land s[i_l]$ for every $s \in R$, which concludes the proof.
\end{proof}

Lemma~\ref{lemma:upp_ie0} also shows that if $R \in \cc{IE}{}{}$ with
a determined argument $i$ then there exists $i_1, \ldots, i_k \in
[\ar(R)]$ such that $t[i] = t[i_1] \land \ldots \land t[i_k]$ for
every $t \in R$, since the constant relation $\{(0)\} \notin
\cc{IE}{}{}$. Before we use Lemma~\ref{lemma:upp_ie0} to show the
non-covering results for $\cc{IE}{}{}$ and $\cc{IE}{}{0}$, we will
need the following lemma, relating the existence of a
upp-definition to a qfpp-definition of a special form. The proof
essentially follows directly from the statement of the lemma and is therefore omitted.

\begin{lemma} \label{lemma:qfpp_to_upp}
  Let $\Gamma$ be a constraint language. Then an $n$-ary relation $R
  \in \cclone{\Gamma}_{\exists !}$ has a upp-definition $R(x_1, \ldots, x_n)
  \equiv \exists ! y_1, \ldots, y_m \colon \varphi(x_1, \ldots, x_n,
  y_1, \ldots, y_m)$ if and only if there exists an $(n+m)$-ary relation $R' \in
  \pcclone{\Gamma}$ such that $\pro_{1, \ldots, n}(R') = R$ where
  each $n < i \leq n+m$ is determined by $1, \ldots, n$. 
\end{lemma}


Say that a partial operation $f$ is {\em $\land$-closed} if
$\domain(f)$ is preserved by $\land$ and that it is {\em 0-closed} if
$\mathbf{0}^{\ar(f)} \in \domain(f)$. We may now describe
partial polymorphisms of $\cclone{\Gamma}_{\exists !}$ using
$\land$-closed and $0$-closed partial polymorphisms of $\Gamma$.

\begin{lemma} \label{lemma:land_closed}
  Let $\Gamma$ be a constraint language such that $\cclone{\Gamma} =
  \cc{IE}{}{0}$. If $f \in \ppol(\Gamma)$ is $\land$- and 0-closed then $f \in
  \ppol(\cclone{\Gamma}_{\exists !})$.
\end{lemma}

\begin{proof}
  Let $R \in \cclone{\Gamma}_{\exists !}$ be an $n$-ary relation and
  let $R' \in \pcclone{\Gamma}$ be the $(n+m)$-ary relation from
  Lemma~\ref{lemma:qfpp_to_upp} where $\pro_{1, \ldots, n}(R') = R$
  and each $i \in \{n+1, \ldots, n+m\}$ is determined in $R'$. Since
  $f$ preserves each relation in $\Gamma$ it follows that $f$ also
  preserves $R'$ since $R' \in \pcclone{\Gamma}$. Assume, with the aim
  of reaching a contradiction, that there exists $s_1, \ldots, s_k \in
  R$ such that $f(s_1, \ldots, s_k) \notin R$ ($\ar(f) = k$), and let $t_1, \ldots,
  t_n$ be the corresponding tuples in $R'$ --- guaranteed to exist due
  to the construction of $R'$. Now, for each $i \in \{n+1, \ldots,
  n+m\}$, either there exists $i_1, \ldots, i_{\ell} \in [n]$ such that
  $t[i] = t[i_1] \land \ldots \land
  t[i_{\ell}]$ for each $t \in \{t_1, \ldots, t_k\}$, or $t_j[i] = 0$ for each $t_j \in \{t_1, \ldots,
  t_k\}$ (from Lemma~\ref{lemma:upp_ie0}). But since $(t_1[i_1], \ldots, t_k[i_1]), \ldots,
  (t_1[i_{\ell}], \ldots, t_k[i_{\ell}]) \in \domain(f)$ (since
  $f(s_1, \ldots, s_k)$ is defined) it follows that $(t_1[i], \ldots,
  t_k[i]) \in \domain(f)$, too, due to the assumption that $f$ is
  $\land$-closed and 0-closed. However, then
  $f(t_1, \ldots, t_k)$ is also defined and it follows that $f(t_1,
  \ldots, t_k) \notin R'$, contradicting the assumption that $f \in \ppol(\Gamma)$.
\end{proof}

We now have all the technical machinery in place to prove that
$\cc{IE}{}{0}$ and $\cc{IE}{}{}$ are not ${\exists !}$-covered.

\begin{theorem} \label{lemma:ie0}
  Let $R_w$ be the weak base of $\cc{IE}{}{0}$ from Table~\ref{table:weak_plain_bases}. Then
  $\cclone{R_w}_{\exists !}
  \subset \cc{IE}{}{0}$.
\end{theorem}

\begin{proof}
  We prove that the relation $R(x_1, x_2, x_3) \equiv x_1
  \leftrightarrow x_2x_3$ is not upp-definable over $R_w$, which is
  sufficient since $R \in
  \cc{IE}{}{0}$, as evident in Table~\ref{table:weak_plain_bases}.
  Furthermore, using Lemma~\ref{lemma:qfpp_to_upp}, we only have to
  prove that
  any $(3+n)$-ary $R'$ where $\pro_{1,2,3}(R') = R$, and where each other argument is determined
  by the three first, is not included in $\pcclone{R_w}$. 
  Assume, without loss of generality,
  that $R'$ does not contain any redundant arguments. Define the binary partial operation $f$
  such that $f(0,0) = 0$, $f(0,1) = f(1,0) = 1$. By construction,
$f$ is both 0-closed and $\land$-closed, and it is also readily
  verified that $f$ preserves $R_w$,  which via
  Lemma~\ref{lemma:land_closed} then implies that $f \in
  \ppol(\cclone{R_w}_{\exists !})$. 
  To finish the proof we also need to show that $f \notin
  \ppol(R')$, which is sufficient since it implies that $R' \notin
  \cclone{R_w}_{\exists !}$. Take two tuples $s, t \in R'$ such that
  $\pro_{1,2,3}(s) = (0,0,1)$, and $\pro_{1,2,3}(t) = (0,1,0)$. From Lemma~\ref{lemma:upp_ie0},
  for each $3 < i \leq n + 3$,  either $i$ is constant 0 in $R'$ or there exists $i_1, \ldots, i_k \in
  \{1,2,3\}$, $k \leq 3$, such that $t[i] = t[i_1] \land \ldots
  t[i_k]$ for each $t \in R'$. But
  then  $(s[i], t[i]) \in \domain(f)$ for each $3 < i \leq n + 3$, since either $(s[i], t[i]) = (0, 0) \in \domain(f)$ or $(s[i], t[i])$ is a conjunction over
  $(0,0,1)$ and $(0,1,0)$. However, this implies that $f(s,t) = u
  \notin R'$ since $\pro_{1,2,3}(u) = (0,1,1)$. Hence, $f$ does not
  preserve $R'$, and $R' \notin \pcclone{R_w}$ via Theorem~\ref{theorem:galois}.
\end{proof}

The proof for $\cc{IE}{}{}$ uses the same construction and we omit the
details. Surprisingly, as we will now see, $\cc{IE}{}{1}$ and
$\cc{IE}{}{2}$ behave entirely differently and are in fact $\exists
!$-covered.

\begin{lemma} \label{lemma:ie1_and_ie2}
  $\cc{IE}{}{1}$ and $\cc{IE}{}{2}$ are $\exists !$-covered.
\end{lemma}

\begin{proof}
  We begin with $\cc{IE}{}{1}$. Let $R_{\cc{IE}{}{1}}(x_1, x_2, x_3, x_4) \equiv (x_1 \leftrightarrow x_2x_3) \land
  \T(x_4)$ be the weak base of $\cc{IE}{}{1}$, and $\Gamma_p = \{(\neg
  x_1 \vee \ldots \vee \neg x_k \vee x) \mid k \geq 0\}$ the plain
  base of $\cc{IE}{}{1}$, from Table~\ref{table:weak_plain_bases}. First, note that for $k
  = 0$ the relation $(\neg x_1 \vee \ldots \vee x_k \vee x) \equiv
  \T(x)$ and may be upp-defined by $\T(x) \equiv R_{\cc{IE}{}{1}}(x, x, x, x)$.
  Second, observe that if
  we can qfpp-define $(\neg x_1 \vee \ldots \vee \neg x_k \vee x)
  \equiv (x_1 \ldots x_k \rightarrow x)$
  then we can also qfpp-define $(x_1 \ldots x_k) \leftrightarrow x$,
  since (1) $(x_1 \rightarrow x_2) \equiv (x_1 \ldots x_1) \rightarrow
  x_2$ and (2) $(x_1 \ldots x_k) \leftrightarrow x \equiv x_1 \ldots x_k
  \rightarrow x \land (x \rightarrow x_1) \land \ldots \land (x
  \rightarrow x_k)$. We may then upp-define $(x_1 \ldots x_k
  \rightarrow x)$ for $k \geq 1$ as

\begin{enumerate}
\item
  $x_1 \rightarrow x \equiv \exists ! y \colon R_{\cc{IE}{}{1}}(x_1, x_1, x, y)$,
\item
  $(x_1x_2 \rightarrow x) \equiv \exists ! x_3, x_4 \colon R_{\cc{IE}{}{1}}(x_3,
  x_1, x_2, x_4) \land x_3 \rightarrow x$, and for $k \geq 3$
\item
  $(x_1 \ldots x_k \rightarrow x) \equiv \exists ! y \colon (x_1 \ldots x_{k-1}
  \leftrightarrow y) \land (x_ky) \rightarrow x$,
  \end{enumerate}
using the upp-definable relation $(x_1 \ldots x_{k-1} \rightarrow x)$
at level $k - 1$.

Let us now consider $\cc{IE}{}{2}$. Let $R_{\cc{IE}{}{2}}(x_1, x_2, x_3, x_4, x_5) \equiv (x_1 \leftrightarrow x_2x_3) \land
  \F(x_4) \land \T(x_5)$ be the weak base of $\cc{IE}{}{2}$, and $\Gamma_p = \{\nandn{}{\ra}(x_1, \ldots, x_{\ra}) \mid \ra
  \in \mathbb{N}\} \cup \{(\neg
  x_1 \vee \ldots \vee \neg x_k \vee x) \mid k \geq 0\} $ the plain
  base of $\cc{IE}{}{2}$, from Table~\ref{table:weak_plain_bases}. Since the
  upp-definitions of $\{(\neg x_1 \vee \ldots \vee \neg x_k \vee x)
  \mid k \geq 0\}$ are similar to the $\cc{IE}{}{1}$ case we only
  present the upp-definitions of $\{\nandn{}{\ra}(x_1, \ldots, x_{\ra}) \mid \ra
  \in \mathbb{N}\}$. First, observe that $R_{\cc{IE}{}{1}}(x_1, x_2,
  x_3, x_4) \equiv \exists ! x_5 \colon R_{\cc{IE}{}{2}}(x_1, x_2,
  x_3, x_5, x_4)$, implying that the relation $(x_1 \ldots x_k
  \rightarrow x)$ is upp-definable over $R_{\cc{IE}{}{2}}$. Then
  $\nandn{}{k}(x_1, \ldots, x_k)$ can be upp-defined as
  $\nandn{}{k}(x_1, \ldots, x_k) \equiv \exists ! x,y \colon (x
  \leftrightarrow x_1\ldots x_k) \land R_{\cc{IE}{}{2}}(x, x, x, x, y)$.
\end{proof}



The natural generalisation of the Boolean operations $\land$ and
$\lor$ are so-called {\em semilattice operations}; binary
operations that are idempotent, associative, and commutative. It is then
tempting to conjecture that Lemma~\ref{lemma:upp_ie0} can be
generalized to arbitrary semilattice operations, i.e., that every
determined argument can be described as a semilattice combination of
other arguments, whenever a relation is preserved by a given
semilattice operation. This, however, is not true. For a simple
counterexample define the semilattice operation $s \colon \{0,1,2\}^2
\rightarrow \{0,1,2\}$ as $s(x,x) = x$ and
$s(x,y) = 0$ otherwise. If we then consider the relation
$R = \{(0,0),(1,1), (2,0)\}$ it is easily verified that $s$ preserves
$R$, and that the second argument is uniquely determined
by the first argument but cannot be described via the
operation $s$. 

The only co-clones remaining are $\cc{IS}{}{11}$ and
$\cc{IS}{n}{11}$ (for $n \geq 2$). As we will
see, unique existential quantification is only as powerful as frozen
quantification for these co-clones. We state the following lemma only
for $\cc{IS}{}{11}$ but the same construction is valid also for
$\cc{IS}{n}{11}$.

\begin{lemma} \label{lemma:is1101_frozen}
  Let $\Gamma$ be a constraint language such that $\cclone{\Gamma} =
  \cc{IS}{}{11}$. Then $\cclone{\Gamma}_{\exists !} = \cclone{\Gamma}_{\mathrm{fr}}$.
\end{lemma}

\begin{proof}
Let $R$ be an $n$-ary relation in $\cc{IS}{}{11}$. Our aim is to prove
that if an argument $i \in [n]$ of $R$ is determined then $i$ is either redundant
or constant 0. This is sufficient to prove the claim since any
upp-definition over $\Gamma$ can then be transformed into an
fpp-definition over $\Gamma$. 

Hence, assume that $i$ is neither redundant nor constant
0. Recall that $\cc{IS}{}{11} \subset \cc{IE}{}{0}$ and that
$\clone{\{x \land (y \lor z), 0\}} = \cc{S}{}{11}$. Lemma~\ref{lemma:upp_ie0} then
implies that such an $i$ can be described as a conjunction of other
arguments, i.e., that there exists
$i_1, \ldots, i_{k} \in [k]$ distinct from $i$ such that $t[i] = t[i_1] \land \ldots
\land t[i_k]$ for every $t \in R$. Note that $k > 1$ as otherwise $i$
is redundant in $R$. Partition $R$ into two sets $R_0$ and $R_1$ such
that $R_0 = \{t \in R \mid t[i] = 0\}$ and $R_1 = \{t \in R \mid t[i]
= 1\}$; both non-empty by our assumption that $i$ is non-constant. By
the assumption that $i$ is determined by a conjunction of $i_1,
\ldots, i_k$ it first follows that there exists $s \in R_1$ where
$s[i_1] = \ldots = s[i_k] = 1$, and that there for each $i_j \in
\{i_1, \ldots, i_k\}$ exists
a tuple $s_{i_j} \in R_0$ such that $s_{i_j}[i_j] \neq s_{i_j}[i] = 0$, as
otherwise $t[i] = t[i_j]$ for each $t \in R$, and $i$ is
redundant. Now, consider an application of the form $s \land (s_{i_j} \lor
s_{i_l}) = s'$ for $i_j, i_l \in \{i_1, \ldots, i_k\}$. Since
$\pro_{i_1, \ldots, i_k}(s) = (1, \ldots, 1)$, we have that $\pro_{i_1, \ldots,
  i_k}(s') = \pro_{i_1, \ldots, i_k}(s_{i_j}) \lor \pro_{i_1, \ldots, i_k}(s_{i_l})$, and since
$s_{i_j}[i] = s_{i_l}[i] = 0$, we also know that $s'[i] = 0$. If we then
consider the application $s' = s \land (s_{i_1} \lor (s \land (s_{i_2} \lor (\ldots
\lor (s \land (s_{i_k - 1} \lor s_{i_k})) \cdots ))))$ it follows that
$\pro_{i_1, \ldots, i_k}(s') = \pro_{i_1, \ldots, i_k}(s)$ since
$\pro_{i_1, \ldots, i_k}(s') = \pro_{i_1, \ldots, i_k}(s_{i_1}) \lor
\ldots \lor \pro_{i_1, \ldots, i_k}(s_{i_k})$, and, furthermore, that
$s' \in R_0$ since $s'[i] = 0$. This contradicts the assumption that
$i$ is determined by $i_1, \ldots, i_k$, and we conclude that $i$ must
be redundant or constant 0 if it is determined in $R$.
\end{proof}

It thus only remains to
prove that $\cc{IS}{}{11}$ and $\cc{IS}{n}{11}$ do not collapse into a
single frozen co-clone. Here, we state the lemma only for
$\cc{IS}{n}{11}$, but the same argument
works for $\cc{IS}{}{11}$.

\begin{lemma} \label{lemma:11}
  Let $\Gamma_p$ denote the plain base and $\Gamma_w$ the weak base of
  $\cc{IS}{n}{11}$ ($n \geq 2$) from Table~\ref{table:weak_plain_bases}.
  Then $\cclone{\Gamma_w}_{\mathrm{fr}} \subset \cclone{\Gamma_p}_{\mathrm{fr}}$.
\end{lemma}

\begin{proof}
  We will show that there exists a partial operation $f$ such that
  $f(\zeroes^n) = 0$, and such that $f$ preserves $\Gamma_w$ but not $\Gamma_p$,
  which is sufficient to prove the claim according to Theorem 12 in
  Nordh \& Zanuttini~\cite{nordh2009}. Let $f$ be defined on $s_1,
  \ldots, s_n \in \{0,1\}^{n}$ such that the Hamming weight of each
  $s_i$ is equal to 1, and such that $f(s_1) = \ldots  = f(s_n) =
  1$. By definition, $f$ does not preserve $\{0,1\}^{n} \setminus
  \{\ones^n\}$, and does therefore not preserve $\Gamma_p$, either. We
  now claim that $f$ preserves $\Gamma_w$. Indeed, consider an
  application $f(t_1, \ldots, t_n)$ for $t_1, \ldots t_n \in
  R_w$. Then either (1) there exists $i$ such that $(t_1[i], \ldots,
  t_{n}[i])$ has Hamming weight larger than 1, or (2) $\{t_1, \ldots,
  t_n\} = \{\zeroes^{n+2}, t_i\}$ for some $i \in [n]$. To see why the
  second case is true, simply note that $t[n+1] = 1$ for every $t \in
  \Gamma_w \setminus \{\zeroes^{n+2}\}$, and if we insist that the
  Hamming weight of each $(t_1[i], \ldots,
  t_{n}[i])$ is smaller than or equal to 1, then the sequence $t_1,
  \ldots, t_n$ contains $n-1$ or $n$ repeated applications of
  $\zeroes^{n+2}$. We can now wrap up the proof by a simple case
  analysis of the tuple $t_i$. First, assume that $t_i =
  \zeroes^{n+2}$. In this case $f(t_1, \ldots, t_n) = \zeroes^{n+2} \in
  \Gamma_w$. Second, assume that $\pro_{1, \ldots, n}(t_i)$ has
  Hamming weight between 1 and $n - 1$ (the case where it is equal to $n$
  is impossible). In this case there exists $j$ such that $(t_1[j],
  \ldots, t_n[j]) = \zeroes^{n}$, implying that $f(t_1, \ldots, t_n)
  = t$ for a tuple where $\pro_{1, \ldots, n}(t) \neq \ones^{n}$, $t[n+1]
  = 1$, and $t[n+2] = 0$, and hence that $t \in \Gamma_w$.
\end{proof}














Combining the results in this section we can now finally prove our
dichotomy theorem.










\begin{theorem} \label{thm:dic}
  Let $\cclone{\Gamma}$ be a Boolean co-clone. Then $\cclone{\Gamma}$
  is not $\exists !$-covered if and only if

\begin{enumerate}
\item
  $\cclone{\Gamma} \in \{\cc{IE}{}{}, \cc{IE}{}{0}, \cc{IV}{}{},
  \cc{IV}{}{1}\}$, or
\item
  $\cclone{\Gamma} \in \{\cc{IS}{n}{01}, \cc{IS}{n}{11} \mid n \geq
  2\} \cup \{\cc{IS}{}{01}, \cc{IS}{}{11}\}$ (where, in addition,
  $\cclone{\Gamma}_{\exists !} = \cclone{\Gamma}_{\mathrm{fr}}$).
\end{enumerate}

\end{theorem}

\begin{proof}
  Each negative case either follows immediately from Lemma~\ref{lemma:ie0},
  Lemma~\ref{lemma:is1101_frozen}, Lemma~\ref{lemma:11}, or is the dual
  of one of those cases. Each $\exists !$-covered co-clone 
  is proved in Lemma~\ref{lemma:simple},
  Theorem~\ref{thm:ess_boolean}, Lemma~\ref{lemma:id2}, and
  Lemma~\ref{lemma:ie1_and_ie2}. 
\end{proof}

\section{Applications in Complexity}
\label{sec:complexity}
In this section we apply Theorem~\ref{thm:dic} to study the complexity
of computational problems not compatible with pp-definitions. Let us
begin by defining the {\em constraint satisfaction problem} over a
constraint language $\Gamma$ ($\CSP(\Gamma)$).

\smallskip
\noindent
{\sc Instance:} A tuple $(V,C)$ where $V$ is a set of variables and
$C$ a set of constraints of the form $R_i(x_{i_1}, \ldots,
x_{i_{\ar(R)}})$ for $R_i \in \Gamma$.

\noindent
{\sc Question:} Does $(V,C)$ have at least one model? That is, a function $f \colon V \rightarrow D$
such that $f(x_{i_1}, \ldots, x_{i_{\ar(R_i)}}) \in R_i$ for each
$R_i(x_{i_1}, \ldots, x_{i_{\ar(R_i)}}) \in C$?
\smallskip

For Boolean constraint languages $\Gamma$ we write $\SAT(\Gamma)$ instead of
$\CSP(\Gamma)$. If $\Delta
\subseteq \cclone{\Gamma}$ (or, equivalently, $\pol(\Gamma) \subseteq
\pol(\Delta)$) then $\CSP(\Delta)$ is polynomial-time reducible to
$\CSP(\Gamma)$~\cite{jeavons1998}. However, there exist many natural
variants of CSPs not compatible with pp-definitions, but compatible
with more restricted closure operators such as upp-definitions. One
such example is the {\em unique satisfiability problem} over a Boolean
constraint language $\Gamma$ ($\USAT(\Gamma)$).



\noindent
\smallskip
{\sc Instance:} A $\SAT(\Gamma)$ instance $I$.

\noindent
{\sc Question:} Does $I$ have a unique model?
\smallskip

The unrestricted $\USAT$ problem, i.e., the $\USAT$ problem where all
possible constraints are allowed, can be seen as the intersection of
satisfiability (in NP), and the satisfiability problem of checking if
a given instance does not admit two distinct models (in co-NP). Hence,
$\USAT$ is included in the second level of the Boolean hierarchy,
BH$_2$, but is not believed to be complete for this
class~\cite{PAPADIMITRIOU1984244}. This unclear status motivated Blass
and Gurevich~\cite{BLASS198280} to introduce the complexity class {\em
unique polynomial-time}, US, the set of decision problems solvable by
a non-deterministic polynomial-time Turing machine where an instance
is a yes-instance if and only if there exists a unique accepting
path. Blass and Gurevich then quickly observed that $\USAT$ is
US-complete and that US $\subseteq$ BH$_2$.


We will present a simple, algebraic proof of Juban's trichotomy
theorem for $\USAT(\Gamma)$~\cite{juban99}, showing that
$\USAT(\Gamma)$ for finite $\Gamma$ is either tractable,
co-NP-complete, or US-complete. Using our machinery we will also be
able to generalise this result to arbitrary infinite constraint
languages. However, for infinite $\Gamma$ we first need to specify a
method of representation. We assume that the elements $R_1, R_2,
\ldots$ of $\Gamma$ are recursively enumerable by their arity, are
represented as lists of tuples, and that there exists a computable
function $f \colon \mathbb{N} \rightarrow \mathbb{N}$ such that for
every $k \geq 1$ and every $k$-ary relation $R$, $R \in
\cclone{\Gamma}_{\exists !}$ if and only if $R \in \cclone{\Gamma \cap
\relations{\{0,1\}}^{\leq f(k)}}_{\exists !}$. Thus, if a relation is
upp-definable it is always possible to bound the arities of the
required relations in the definition.
The complexity of
$\USAT(\Gamma)$ is then determined by $\cclone{\Gamma}_{\exists !}$ in
the following sense.

\begin{theorem} \label{thm:usat_complexity}
  Let $\Gamma$ and $\Delta$ be Boolean constraint
  languages. If $\Delta \subseteq \cclone{\Gamma}_{\exists !}$ is
  finite then $\USAT(\Delta)$ is polynomial-time many-one reducible to $\USAT(\Gamma)$.
\end{theorem}

\begin{proof}
  By assumption every $R \in \Delta$ is upp-definable over
  $\Gamma$. First let $k = \max \{f(\ar(R)) \mid R \in \Delta\}$. We
  then begin by computing a upp-definition of $R$
  over $\Gamma \cap \relations{\{0,1\}}^{\leq k}$, and
  store this upp-definition in a table. Since $\Delta$ is finite this
  can be done in constant time. Next, given an instance $I = (V, C)$
  of $\USAT(\Delta)$, we similar to the ordinary $\CSP$ case simply
  replace each constraint in $C$ by its upp-definition over $\Gamma$,
  and identify any potential variables occurring in equality
  constraints. This procedure might introduce additional variables,
  but since they are all determined by $V$, the existence of a unique
  model is preserved.
\end{proof}

\begin{theorem} \label{thm:usat}
  Let $\Gamma$ be a Boolean constraint language. Then $\USAT(\Gamma)$
  is co-NP-complete if $\cclone{\Gamma} \in \{\cc{II}{}{0},
  \cc{II}{}{1}\}$, US-complete if $\cclone{\Gamma} = \cc{II}{}{2}$, and
  is tractable otherwise.
\end{theorem}

\begin{proof}
  We begin with the tractable cases and assume that $\cclone{\Gamma}
\notin \{\cc{II}{}{0}, \cc{II}{}{1}, \cc{II}{}{2}\}$. If
$\cclone{\Gamma} = \cc{IN}{}{2}$ or $\cclone{\Gamma} = \cc{IN}{}{}$ then any instance of $\USAT(\Gamma)$
is either unsatisfiable or has at least two models, since any $R \in
\cc{IN}{}{} \subset \cc{IN}{}{2}$ is closed under complement. Then, trivially,
$\USAT(\Gamma)$ can be solved in constant time by always
answering no. Similarly, if $\cclone{\Gamma} = \cc{II}{}{}$ then any
instance of $\USAT(\Gamma)$ has at least two models, and we simply
answer no. Every other case can then be solved efficiently by
enumerating models wih polynomial delay~\cite{schnoor2007}, by 
answering no if more than one model is found.


For the intractable cases assume first that $\cclone{\Gamma} =
\cc{II}{}{0}$. Membership in co-NP is clear since a no-instance can be
verified by any non-constant model. Let $\UNSAT(\Gamma)$ denote the
unsatisfiability problem over a Boolean constraint language $\Gamma$,
and consider the problem $\UNSAT(\{R\})$ where $R = \{(0,0,1,0,1),
(0,1,0,0,1), (1,0,0,0,1)\}$. It is readily seen that $\cclone{R} =
\cc{II}{}{2}$ since $R$ is only preserved by projections, implying that $\UNSAT(\{R\})$ is co-NP-complete, and we
will show co-NP-hardness of $\USAT(\Gamma)$ by a
polynomial-time many-one reduction from $\UNSAT(\{R\})$ to $\USAT(\{R
\cup \{(0,0,0,0,0)\}\})$. Given an instance $(V,C)$ of $\UNSAT(\{R\})$
we begin by introducing one fresh variable $c_1$, and for each
variable $x^5_i$ occurring in a constraint $R(x^1_i, x^2_i, x^3_i, x^4_i,
x^5_i) \in C$ we replace $x^5_i$ with $c_1$. Then, we for each constraint
$R(x^1_i, x^2_i, x^3_i, x^4_i,
c_1)$ replace it by $(R \cup \{(0,0,0,0,0)\})(x^1_i, x^2_i, x^3_i, x^4_i,
c_1)$, and let $I'$ be the resulting instance of $\USAT(\{R \cup
\{(0,0,0,0,0)\}\})$. The fresh variable $c_1$ which occurs in every
constraint ensures that if a
constraint is satisfied by assigning all variables 0, then all other
variables have to be assigned 0 as well. It follows that $I'$ admits a unique model,
namely the model where each variable is assigned 0, if and only if $I$
is unsatisfiable. Theorem~\ref{thm:dic} and Theorem~\ref{thm:usat_complexity}
then gives co-NP-hardness for every other base $\Gamma$ of $\cc{I}{}{0}$ for
free. The case when $\cclone{\Gamma} = \cc{I}{}{1}$ is entirely
analogous and we therefore omit it.

Last, assume that $\cclone{\Gamma} = \cc{II}{}{2}$. Let $\UKSAT{k}$ denote
the $\USAT$ problem restricted to constraints formed by $k$-ary
clauses, and recall that $\USAT$ is US-complete. Following the
succinct upp-definition provided in Example~\ref{ex:upp} we may then
conclude that $\UKSAT{k}$ is also US-complete for every $k \geq 3$. Since
each $k$-ary clause can be represented by a $k$-ary relation,
Theorem~\ref{thm:dic} and Theorem~\ref{thm:usat_complexity} then shows
US-completeness for every base $\Gamma$ of $\cc{II}{}{2}$.
\end{proof}

A complexity classification akin to Theorem~\ref{thm:usat} is useful
since it clearly separates tractable from intractable cases. However,
in the last decade, a significant amount of research has been devoted
to better understanding the ``fine-grained'' complexity of 
intractable problems, with a particular focus on ruling out algorithms
running in $O(c^{|V|})$ time for every $c > 1$, so-called {\em
  subexponential time}. This line of research originates from
Impagliazzo et al.~\cite{impagliazzo2001} who conjectured that 3-SAT
is not solvable in subexponential time; a conjecture known as the {\em
exponential-time hypothesis} (ETH). Lower bounds for $\USAT(\Gamma)$
can then be proven using the ETH and the
results from Section~\ref{sec:power}.



\begin{theorem} \label{thm:usat_lower}
  Let $\Gamma$ be a Boolean constraint language such that
  $\USAT(\Gamma)$ is US-complete or co-NP-complete.
  Then $\USAT(\Gamma)$ is not
  solvable in subexponential time, unless the ETH is false.
\end{theorem}

\begin{proof}
  We begin with the case when $\cclone{\Gamma} = \cc{II}{}{0}$ or $\cclone{\Gamma} = \cc{II}{}{1}$.  First, observe that $\USAT(\Gamma)$ is 
  the complement of determining whether there exists a
  non-constant model. The latter problem, in turn, can be seen as a
  special case of the problem of determining if there exists a
  surjective model, and is in the literature referred to as
  SUR-SAT$(\Gamma)$ or $\SAT^{*}(\Gamma)$~\cite{creignou2008b}. It is furthermore known
  that if SUR-SAT$(\Gamma)$ for $\cclone{\Gamma} \in \{\cc{II}{}{0}, \cc{II}{}{1}\}$ is solvable in subexponential time then
  the ETH is false~\cite{jonssonetal2014}[Section 5].

  Second, assume that $\cclone{\Gamma} = \cc{II}{}{2}$. Assume that $\USAT(\Gamma)$ is solvable in subexponential
  time. Results from Jonsson et al.~\cite{jonsson2017} then imply that
  $\USAT(\Rdddp)$, where \[\Rdddp = \{(0,0,1,1,1,0,0,1),
  (0,1,0,1,0,1,0,1), (1,0,0,0,1,1,0,1)\},\] is solvable in
  subexponential time, too. It is furthermore known that the ETH is
  false if and only if $\SAT(\Rdddp)$ restricted to instances
  containing at most $2|V|$ constraints, is solvable in subexponential
  time~\cite{jonsson2017}. Hence, it suffices to show that the original assumption
  implies that $\SAT(\Rdddp)$, restricted to instances with at most
  $2|V|$ constraints, is solvable in subexponential time, in
  order to contradict the ETH. 
Define the $9$-ary relation $R_{\lor}$
  as $R_{\lor} = \{(x_1, \ldots, x_8, b) \mid (x_1, \ldots, x_8)
  \in \Rdddp, b \in \{0,1\}\} \cup \{(x_1, \ldots, x_8, 1) \mid (x_1,
  \ldots, x_8) \notin \Rdddp\}$. Let \[R_{\lor}(x_1, \ldots, x_8, x_9)
  \equiv \exists ! y_1, \ldots, y_{D} \colon \varphi_{R_{\lor}}(x_1,
  \ldots, x_8, x_9, y_1, \ldots, y_{D})\] be a upp-definition of
  $R_{\lor}$ over $\Rdddp$; this is possible due to
  Theorem~\ref{thm:dic}. Similarly, let $(x_1 \rightarrow x_2) \equiv
  \exists ! y_1, \ldots, y_{E} \colon \varphi_{x_1 \rightarrow
    x_2}(x_1, x_2, y_1, \ldots, y_E)$ be a upp-definition of $(x_1
  \rightarrow x_2)$ over $\Rdddp$. The reduction then proceeds as follows. Let $I = (V,C)$ be an instance of $\SAT(\Rdddp)$ where $|C| \leq
  2|V|$. 
  We introduce one fresh variable $x$ together with the
  constraints $\bigwedge^{|V|}_{i = 1} \varphi_{\rightarrow}(x, x_i,
  y^1_i, \ldots, y^{E}_i)$, where $y^{1}_1, \ldots, y^E_1, \ldots,
  y^1_n, \ldots, y^E_n$ are fresh variables. For each constraint $c_i = \Rdddp(x^1_i,
  \ldots, x^8_i)$ we then replace it by $\varphi_{R_{\lor}}(x^1_i,
  \ldots, x^8_i, x, z^{1}_{i}, \ldots, z^{D}_i)$, where $z^1_i \ldots,
  z^D_i$ are fresh variables. 

  Let $I' = (V', C')$ be the resulting instance of $\USAT(\Rdddp)$,
  where $V' = V \cup \{x\} \cup Y$ such that $Y$ consists of the 
  variables introduced when replacing constraints in $C$ by their
  upp-definitions over $\Rdddp$. We claim that
  $I$ is unsatisfiable if and only if $I'$ admits a unique
  model. First assume that $I$ is unsatisfiable. In this case any
  model $f$ of $I'$ must satisfy $f(x) = 1$. However, due to the
  constraints $\bigwedge^{|V|}_{i = 1} \varphi_{\rightarrow}(x, x_i,
  y^1_i, \ldots, y^{E}_i)$ this also implies that $f(x_i) = 1$ for
  each $x_i \in V$. But since each $y \in Y$ is determined by a set of
  variables in $V$, it follows that $I'$ has a unique model. For the
  other direction, assume that $I'$ has a unique model $f$. Assume
  first that $f(x) = 1$. Then the constraints $\bigwedge^{|V|}_{i = 1} \varphi_{\rightarrow}(x, x_i,
  y^1_i, \ldots, y^{E}_i)$ force $f(x_i) = 1$ for each $x_i \in V$,
  and it is trivial to verify that $f$ does not satisfy $I$, and that
  the existence a model of $I$ would contradict the uniqueness of
  $f$. Second, assume that $f(x) = 0$. Define $g \colon V' \rightarrow
  \{0,1\}$ such that $g(x) = 1$, $g(x_i) = 1$ for each $x_i \in V'$,
  and $g(y_i)$ for $y_i \in Y$ according to the value prescribed by
  $g$ in the constraint containing $y_i$. This is possible since each
  variable in $Y$ is determined by $V \cup \{x\}$. However, then $g$ is also a model of
  $I'$, contradicting the uniqueness assumption of $f$.

  The above construction can clearly be carried out in polynomial time
  with respect to $|V|$ and $|C|$. For the time complexity, the constraints $\bigwedge^{|V|}_{i = 1} \varphi_{\rightarrow}(x, x_i,
  y^1_i, \ldots, y^{E}_i)$ require $|V| \cdot E$ fresh variables, and
  the constraints $\bigwedge^{|C|}_{i = 1}\varphi_{R_{\lor}}(x^1_i,
  \ldots, x^8_i, x, z^{1}_{i}, \ldots, z^{D}_i)$ require $|C| \cdot D$
  fresh variables, which is bounded by $2|V|$. Hence, $|V'|$ is
  bounded by $|V| + |V|E + 2|V|D$, and 
  since $D$ and $E$ are both constant it follows that $\SAT(\Rdddp)$ is solvable in
  subexponential time, which contradicts the ETH.
\end{proof}

Using our algebraic framework, hardness results can effortlessly be proven
for the $\CSP$ generalisation of $\USAT$, i.e., the problem
$\UCSP(\Gamma)$ of answering yes if and only if the given instance
of $\CSP(\Gamma)$ admits a unique model.

\begin{theorem} \label{thm:ucsp_complexity}
  Let $\Gamma$ be a constraint language over a finite domain $D$. If
  $\cclone{\Gamma} = \relations{D}$ then $\UCSP(\Gamma)$ is
  US-complete, and if $\pol(\Gamma) = \clone{\{f\}}$ for a constant
  operation $f$, then $\UCSP(\Gamma)$ is co-NP-complete.
\end{theorem}

\begin{proof}
  First, assume that $\cclone{\Gamma} = \relations{D}$. Inclusion in
  US follows directly from the definition of $\UCSP(\Gamma)$. To show hardness we take a Boolean $\Delta$ such that
  $\USAT(\Delta)$ is US-complete, and since $\cclone{\Delta}_{\exists
    !} \subseteq \cclone{\Gamma}_{\exists !}$ by
  Theorem~\ref{thm:ess}, we can perform a standard gadget reduction
  from $\USAT(\Delta)$ by
  replacing constraints by their upp-definitions over $\Gamma$.

  Second, assume that $\pol(\Gamma) = \clone{\{f\}}$ for a constant
  operation $f$, and let $\{d\}$ be the image of $f$. Then a
  no-instance of $\UCSP(\Gamma)$ can be verified by any model distinct
  from the constant model where each variable is assigned $d$, implying that
  $\UCSP(\Gamma)$ is included in co-NP. The hardness argument is
  similar to Theorem~\ref{thm:usat} and we only provide a sketch of
  the construction. Let $d_1, \ldots, d_k$ be an enumeration of $D$
  and let $R = \{(0,0,1), (0,1,0), (1,0,0)\} \times \{(d_1, \ldots,
  d_k)\}$. Then the problem of checking whether a $\CSP(R)$ instance
  does not admit a model, $\UNCSP(R)$, is co-NP-complete, and we 
  then reduce $\UNCSP(R)$ to $\UCSP(R \cup \{(d, \ldots, d)\})$
  by replacing each constraint by the corresponding constraint over $R
  \cup \{(d, \ldots, d)\}$. Since $R \cup \{(d, \ldots, d)\} \in
  \cclone{\Gamma}_{\exists !}$ by Theorem~\ref{thm:ess}, co-NP-hardness carries over to
  $\UCSP(\Gamma)$.
\end{proof}

\section{Concluding Remarks and Future Research}
\label{sec:conclusions}

We have studied unique existential quantification in pp-definitions,
with a particular focus on finding constraint languages where
existential quantification and unique existential quantification 
coincide. In general, this question appears highly challenging, but
we have managed to find several broad classes of languages where this
is true, and established a complete dichotomy theorem in the Boolean
domain. We also demonstrated that upp-definitions can be applied to
obtain complexity theorems for problems in a more systematic manner
than what has earlier been possible. 
Many interesting open question hinge on the possibility of finding an
algebraic characterisation of upp-closed sets of relations. For
example, it would be interesting to determine the cardinality of the
set $\{\cclone{\Gamma}_{\exists !} \mid \Gamma \subseteq
\cc{II}{}{2}\}$, and hopefully describe all such upp-closed sets. By our classification theorem it suffices to
investigate the Boolean co-clones that are not
$\exists!$-covered, but even this question appears difficult to
resolve using only relational tools. Similarly, a continued description
of the $\exists !$-covered co-clones over finite domains would be
greatly simplified by an algebraic characterisation. Thus, given a set
of relations $\Gamma$, what is the correct notion of a
``polymorphism'' of a upp-definable relation over $\Gamma$? This
question also has a strong practical motivation: essentially all
complexity classifications for CSP related problems over non-Boolean
domain require stronger algebraic tools than pp-definitions, and this
is likely the case also for problems
that can be studied with upp-definitions.


  Another interesting topic is the following computational problem
  concerning upp-definability. Fix a constraint language $\Gamma$, and let $R$ be a
  relation. Is it the case that $R$ is upp-definable over $\Gamma$?
  The corresponding problem for pp-definitions is tractable for
  Boolean constraint languages $\Gamma$~\cite{creignou2008} while the corresponding
  problem for qfpp-definitions is
  co-NP-complete~\cite{kavvadias98,lagerkvist207f}. Note that if
  $\cclone{\Gamma}$ is $\exists !$-covered (which can be checked in
  polynomial time) then $R \in \cclone{\Gamma}_{\exists !}$ can be
  answered by checking whether $R \in \cclone{\Gamma}$. Thus, only the
  co-clones that are not $\exists !$-covered would need to be
  investigated in greater detail.

  Last, it is worth remarking that our notion of uniqueness
  quantification in pp-definitions is not the only one
  possible. Assume that we in $\exists ! x_i \colon R(x_1, \ldots,
  x_i, \ldots, x_n)$ over a domain $D$ do not require that $x_i$ is determined by $x_1,
  \ldots, x_{i-1}, x_{i+1}, \ldots, x_n$ but instead simply obtain the
  relation $\{(d_1, \ldots, d_{i-1}, d_{i+1}, \ldots, d_n) \mid \exists ! d_i
  \in D$ such that $(d_1, \ldots, d_{i-1}, d_i, d_{i+1}, \ldots, d_n) \in
  R)\}$. This notion of unique existential quantification is in
  general {\em not} comparable to existential quantification, since if
  we e.g.\ let $R = \{(0,0), (0,1), (1,0)\}$ then $\T(x) \equiv
  \exists ! y \colon R(y,x)$ even though $\T \notin \cclone{R}$, i.e.,
  is not even pp-definable by $R$ (where
  $T = \{(1)\}$). Thus, it would be
  interesting to determine the resulting closed classes of relations
  and see in which respect they differ from the ordinary co-clone lattice.

\subsection*{Acknowledgements}
We thank Andrei Bulatov for helpful discussions concerning
the topic of the paper.

\bibliographystyle{plain}
\bibliography{references}

\end{document}